\documentclass[12pt]{article}
\usepackage[usenames]{xcolor}
\usepackage{xifthen}

\def\fontsettingup{2} 

\usepackage{amsmath, amsthm, amssymb}
\usepackage[T1]{fontenc}
\ifthenelse{\fontsettingup = 1}{ \usepackage{eulerpx, eucal, tgpagella}   }{}
\ifthenelse{\fontsettingup = 2}{ \usepackage{mathpazo, eufrak, tgpagella} }{}

\usepackage{hyperref}
\hypersetup{
  colorlinks=true,
  linkcolor=blue,
  filecolor=magenta,
  urlcolor=cyan,
}

\usepackage[linesnumbered,boxed,ruled,vlined]{algorithm2e}
\usepackage{tikz}
\usepackage{cleveref}

\usepackage[letterpaper]{geometry}
\newgeometry{
textheight=9in,
textwidth=6.5in,
top=1in,
headheight=14pt,
headsep=25pt,
footskip=30pt,
}

\newtheorem{theorem}{Theorem}

\newtheorem{claim}[theorem]{Claim}
\newtheorem*{claim*}{Claim}
\newtheorem{condition}[theorem]{Condition}
\newtheorem{example}[theorem]{Example}

\newtheorem{lemma}[theorem]{Lemma}
\newtheorem{proposition}[theorem]{Proposition}
\newtheorem{corollary}[theorem]{Corollary}
\theoremstyle{definition}

\newtheorem{definition}[theorem]{Definition}
\newtheorem{remark}[theorem]{Remark}
\newtheorem*{remark*}{Remark}
\newtheorem{assumption}{Assumption}
\newtheorem{conjecture}[theorem]{Conjecture}

\ifthenelse{\fontsettingup = 1}{
  \def\*#1{\mathbf{#1}} 
  \def\+#1{\mathcal{#1}} 
  \def\-#1{\mathrm{#1}} 
  \def\^#1{\mathbb{#1}} 
  \def\!#1{\mathfrak{#1}} 
}{}

\ifthenelse{\fontsettingup = 2}{
  \def\*#1{\boldsymbol{#1}} 
  \def\+#1{\mathcal{#1}} 
  \def\-#1{\mathrm{#1}} 
  \def\^#1{\mathbb{#1}} 
  \def\!#1{\mathfrak{#1}} 
}{}

\def\oPr{\mathbf{Pr}}
\renewcommand{\Pr}[2][]{ \ifthenelse{\isempty{#1}}
  {\oPr\left[#2\right]}
  {\oPr_{#1}\left[#2\right]} } 

\def\oE{\mathbb{E}}
\newcommand{\E}[2][]{ \ifthenelse{\isempty{#1}}
  {\oE\left[#2\right]}
  {\oE_{#1}\left[#2\right]} }

\def\oVar{\mathrm{Var}}
\newcommand{\Var}[2][]{ \ifthenelse{\isempty{#1}}
  {\oVar\left[#2\right]}
  {\oVar_{#1}\left[#2\right]} }

\def\oEnt{\mathbf{Ent}}
\newcommand{\Ent}[2][]{ \ifthenelse{\isempty{#1}}
  {\oEnt\left[#2\right]}
  {\oEnt_{#1}\left[#2\right]} }

\newcommand{\DTV}[2]{\-D_{\mathrm{TV}}\left({#1},{#2}\right)}

\newcommand{\e}{\mathrm{e}}
\renewcommand{\epsilon}{\varepsilon}
\renewcommand{\emptyset}{\varnothing}
\newcommand{\norm}[1]{\left\Vert#1\right\Vert}
\newcommand{\set}[1]{\left\{#1\right\}}
\newcommand{\tuple}[1]{\left(#1\right)} 
\newcommand{\inner}[2]{\left\langle #1,#2\right\rangle}
\newcommand{\tp}{\tuple}
\newcommand{\ol}{\overline}
\newcommand{\abs}[1]{\left\vert#1\right\vert}

\newcommand{\ftp}[1]{\left\lfloor#1\right\rfloor}

\newcommand{\Cor}{\Psi^{\mathrm{sym}}}
\newcommand{\cor}{\Psi^{\-{cor}}}
\newcommand{\Inf}{\Psi}
\newcommand{\diag}{\-{diag}}

\newcommand{\minor}[1]{{\color{red} #1}}

\title{Spectral Independence Beyond Total Influence\\
 on Trees and Related Graphs}

\author{Xiaoyu Chen\thanks{State Key Laboratory for Novel Software Technology, New Cornerstone Science Laboratory, Nanjing University, 163 Xianlin Avenue, Nanjing, Jiangsu, China. \textnormal{E-mails: \url{chenxiaoyu233@smail.nju.edu.cn}, \url{yinyt@nju.edu.cn}, \url{zhangxy@smail.nju.edu.cn}}}
\and 
Xiongxin Yang\thanks{School of Information Science and Technology, Northeast Normal University, 2555 Jingyue Street, Changchun, Jilin, China. \textnormal{E-mail: \url{yangxx500@nenu.edu.cn}}}
\and
Yitong Yin\footnotemark[1]
\and 
Xinyuan Zhang\footnotemark[1]
}

\date{}

\begin{document}
\maketitle
\begin{abstract}
We study how to establish  \emph{spectral independence}, a key concept in sampling, without relying  on total influence bounds,
by applying an \emph{approximate inverse} of the influence matrix.
Our method gives constant upper bounds on spectral independence for two foundational  Gibbs distributions known to have unbounded total influences: 
  \begin{itemize}
  \item The monomer-dimer model on graphs with large girth (including trees).
     Prior to our work, such results were only known for graphs with constant maximum degrees or infinite regular trees, 
     as shown by Chen, Liu, and Vigoda (STOC '21).
  \item The hardcore model on trees with fugacity $\lambda < \mathrm{e}^2$.
    This  remarkably surpasses the well-known $\lambda_r>\e-1$ lower bound for the reconstruction threshold on trees,
    significantly improving upon the current threshold $\lambda < 1.3$,
    established in a prior work by Efthymiou, Hayes, \v{S}tefankovi\v{c}, and Vigoda (RANDOM '23).
  \end{itemize}
Consequently, we establish optimal $\Omega(n^{-1})$  spectral gaps of the Glauber dynamics  for these models on arbitrary trees, 
regardless of the maximum degree $\Delta$.
%
%
 \end{abstract}

\section{Introduction}
%
Sampling from the Gibbs distribution is a fundamental computational problem that has drawn significant interest across various fields, including mathematics, statistical physics, and computer science.
Typical approaches to this problem involve the use of \emph{Markov chain Monte Carlo} (MCMC) methods.
In these methods,  a Markov chain is designed to converge to the desired stationary distribution,
with the goal of rapid convergence, ideally with a mixing time that is polynomially bounded or even near-linear.
However, despite the widespread use of MCMC methods, analyzing the mixing time of Markov chains remains a formidable challenge.
%

To tackle this challenge, an important tool called \emph{spectral independence} was introduced in a seminal work~\cite{anari2020spectral}. This method builds upon recent advancements in the theory of \emph{high-dimensional expanders}~\cite{anari2019logconcave,cryan2019modified,alev2020improved}.

For a distribution $\mu$ over $2^{[n]}$, the influence matrix $\Inf_\mu$ quantifies the correlation between variables and is defined as
\begin{align*}
\forall i,j \in [n], \quad \Inf_\mu(i,j) :=
\begin{cases}
    \Pr[\mu]{j \mid i} - \Pr[\mu]{j \mid \overline{i}} & \text{ if } \Pr[\mu]{i} \in (0, 1),\\
    0 & \text{ otherwise.}
\end{cases}
\end{align*}
The distribution $\mu$ is said to be \emph{$\eta$-spectrally independent} if the maximum eigenvalue $\lambda_{\max}(\Inf_\mu)$ is upper bounded by some $\eta > 0$.
This concept of spectral independence has become highly influential since its introduction to the field.
It has proven to be a powerful tool  for analyzing the mixing time of Markov chains~\cite{anari2019logconcave,cryan2019modified,anari2020spectral,chen2020rapid,chen2021optimal,chen2021rapid,chen2022holant,chen2022optimal,anari2022entropic,chen2022localization,chen2023nearlinear,chen2023strong}.
Moreover, this concept tightly captures the local expansion of high-dimensional expanders~\cite{anari2020spectral} and is intrinsically connected to other well-known mixing properties, such as decay of correlation~\cite{anari2020spectral,chen2020rapid,chen2021optimal,chen2023strong}, coupling~\cite{liu2021coupling,feng2022rapid,blanca2022mixing,chen2023nearlinear,chen2023matching,chen2023strong}, and zero-freeness~\cite{anari2021fractionally,chen2022holant}. 
The notion of spectral independence has also stimulated the development  of several new approaches for analyzing the mixing time of Markov chains, including entropic independence~\cite{anari2022entropic}, field dynamics~\cite{chen2021rapid}, and localization schemes~\cite{chen2022localization}.



However, establishing spectral independence for Gibbs distributions is far from trivial.
In recent years, numerous studies have delved into this endeavor, developing a variety of approaches including correlation decay, stability of polynomial, coupling of local Markov chain, recursive coupling, and (matrix) trickle-down~\cite{chen2020rapid,chen2022holant,liu2021coupling,blanca2022mixing,chen2023nearlinear,oppenheim2018local,abdolazimi2021matrix}.
While these methods have provided valuable insights, it is worth noting that most of them focus on establishing upper bounds for the \emph{total influence}, represented by the infinite norm of influence matrix:
\[
\norm{\Psi_\mu}_\infty := \max_i \sum_j \abs{\Psi_\mu(i,j)},
\]
rather than directly bounding the  maximum eigenvalue $\lambda_{\max}(\Psi_\mu)$. 
This focus on the total influence may not always yield accurate estimates of the spectral independence $\lambda_{\max}(\Psi_\mu)$, especially in critical scenarios.

There are foundational classes of Gibbs distributions conjectured to be spectrally independent, yet proving this conjecture is challenging due to their unbounded total influence and the difficulties in applying more direct approaches for spectral independence. 
Notable examples include the {monomer-dimer model} (matchings) and the {hardcore model} (independent sets) on trees in certain non-uniqueness regimes~\cite{chen2020rapid,chen2021optimal,efthymiou2023optimal}.

It is then important to ask the following question:

\vspace{1ex}
{\centering \it
How can we establish the spectral independence when the total influence is unbounded?
\par}
\vspace{1ex}

One notable method for upper bounding $\lambda_{\max}(\Psi_\mu)$ without relying on total influence is the (matrix) trickle-down method \cite{oppenheim2018local, abdolazimi2021matrix, wang2023sampling}.
In this method, matrix upper bounds $\set{M_{\mu^x}}$ (with $\Psi_{\mu^x} \preceq M_{\mu^x}$) are inductively constructed for the influence matrices $\set{\Psi_{\mu^x}}$ of conditional distributions~$\mu^x$,
with the hope that the resulting upper bound $\lambda_{\max}(M_\mu)\ge \lambda_{\max}(\Psi_\mu)$ is sufficiently tight and relatively easy to analyze.
While this method has proven useful for obtaining spectral independence without total influence bounds, 
particularly for matroid bases~\cite{oppenheim2018local} and edge-colorings \cite{abdolazimi2021matrix, wang2023sampling},
a significant drawback is the highly non-intuitive construction of the upper bound matrices $\{M_{\mu^x}\}$, making it challenging to apply  generally.

On the other hand, when the graphical model becomes a tree,
recent work~\cite{efthymiou2023optimal} introduced a novel inductive approach to bounding  the spectral independence for the hardcore model on trees. 
Compared to other approaches, this inductive method is more intuitive to work with, while directly upper bounding the spectral independence without relying on the total influence bounds.
However, as noted in~\cite{efthymiou2023optimal}, this approach heavily relies on the inductive structure inherent in trees, thereby limiting the potential for generalization to  graphs with cycles.


In this paper, we apply a new direct approach for establishing spectral independence, 
based on an \emph{approximate inverse}  of the influence matrix.
The method is particularly intuitive on trees while exhibiting promising potential for generalization to non-trees.

We instantiate our method on the monomer-dimer model and the hardcore model.
The Gibbs distributions of the monomer-dimer model and the hardcore model on an undirected graph $G=(V,E)$ with a fugacity $\lambda > 0$ are defined as follows:
\begin{itemize}
\item The Gibbs distribution $\mu$ of the \emph{monomer-dimer model} on graph $G$ with fugacity $\lambda$, is supported over the set of matchings in graph $G$. The measure of each matching $M$ is given by $\mu(M) \propto \lambda^{\abs{M}}$.
\item The Gibbs distribution $\mu$ of the \emph{hardcore model} on graph $G$ with fugacity $\lambda$, is supported over the set of independent sets in graph $G$. The measure of each independent set $I$ is given by $\mu(I) \propto \lambda^{\abs{I}}$.
\end{itemize}

\paragraph{The monomer-dimer model.}
%
The mixing time for the monomer-dimer model stands as a foundational problem within MCMC theory.
 The seminal works of Jerrum and Sinclair \cite{jerrum1989approximating, sinclair1988randomised, jerrum1996markov, jerrum2003counting} established an $O(n^2m\log n)$ mixing time for the monomer-dimer model with $n$ vertices and $m$ edges.
   The only notable improvement since then has been  the $O(\exp(\Delta^{O(1)})\cdot m \log m)$ mixing time for graphs with bounded maximum degree $\Delta$,
   achieved in a breakthrough of Chen, Liu, and Vigoda~\cite{chen2021optimal}.
   It is widely believed that establishing spectral independence is a key step towards settling the problem~\cite{chen2021optimal,liu2023spectral}.

As noted in~\cite{chen2021optimal} and~\cite{liu2023spectral},
the monomer-dimer model on the infinity $\Delta$-regular tree $\mathbb{T}_{\Delta}$ has the total influence $\norm{\Psi_\mu}_\infty=\Theta(\sqrt{\lambda \Delta})$, which may become arbitrarily large as the degree $\Delta$ grows.
However, in contrast, the spectral independence remains bounded as $\lambda_{\max}(\Psi_\mu)=O_{\lambda}(1)$ on $\mathbb{T}_{\Delta}$,  regardless of $\Delta$. 
Furthermore, the following was conjectured.
\begin{conjecture}[\cite{chen2021optimal}]\label{conjecture:monomer-dimer-SI}
For any graph $G=(V,E)$ and any $\lambda>0$,
the Gibbs distribution $\mu$ of the monomer-dimer model on $G$ with fugacity $\lambda$ has 
$\lambda_{\max}(\Psi_\mu)=O_{\lambda}(1)$.
\end{conjecture}
\noindent
This conjecture aligns with the renowned $\widetilde{O}(n^2m)$ mixing time bounds in~\cite{jerrum1989approximating,jerrum2003counting}, where the additional factor $n^2$ may be due to a limitation of the canonical path approach.


Inspired by~\cite{liu2023spectral}  and several previous works such as~\cite{chen2021optimal,bapat2013product}, 
we calculate the spectral independence $\lambda_{\max}(\Psi_\mu)$ by analyzing the inverse $\Psi_\mu^{-1}$ of the influence matrix.
When the graph $G$ is a tree, the influence matrix $\Psi_\mu$ becomes a \emph{product distance matrix} on this tree, and its inverse  $\Psi_\mu^{-1}$ reveals the underlying structure of $G$.
%
Then, by examining the ``local influence matrices'' (see~\Cref{thm:Psi-lambda-max} and~\Cref{lem:edge-boundedness}), 
we gain a much clearer understanding of  spectral independence on arbitrary trees.

We prove the following result, confirming \Cref{conjecture:monomer-dimer-SI} on trees.
Consequently, this spectral independence bound also implies an optimal spectral gap for the \emph{Glauber dynamics},
which is a canonical Markov chain for sampling from the Gibbs distribution.

\begin{theorem} \label{thm:matching-tree}
  Let $T = (V, E)$ be a tree of $n$ vertices and  $\lambda > 0$.
  The Gibbs distribution $\mu$  of the monomer-dimer model on $T$ with fugacity $\lambda$ has 
  \[
  \lambda_{\max}(\Psi_\mu) \leq 2\lambda + 1.
  \]
  Moreover, the Glauber dynamics on $\mu$ has asymptotically optimal spectral gap $\Omega_\lambda(n^{-1})$.
\end{theorem}

Indeed, the spectral independence bound of \Cref{thm:matching-tree} for trees follows as a special case from a more general theorem stated next. 
For general graphs with cycles, the inverse $\Psi_\mu^{-1}$ may no longer possess the simple structure as a product distance matrix on a tree.
Therefore, we turn to using an \emph{approximate} inverse of the influence matrix $\Psi_\mu$, which effectively approximates the maximum eigenvalue while preserving the structure of the underlying graph.
The result is stated in the following theorem, which establishes a general trade-off between the girth and the spectral independence in the monomer-dimer model.

\begin{theorem} \label{thm:matching-SI-large-girth}
  Let $G=(V,E)$ be a graph with maximum degree $\Delta$ and girth $g$ for  $g\in[3,\infty]$.
  The Gibbs distribution $\mu$  of the monomer-dimer model on $G$ with fugacity $\lambda > 0$ has the spectral independence
  \[
  \lambda_{\max}(\Inf_\mu) \leq (2\lambda + 1)\tp{4\tp{\sqrt{1 + \lambda\Delta} + 1}\tp{1-\frac{2}{\sqrt{1+\lambda\Delta}+1}}^{\left\lfloor(g-1)/4\right\rfloor} + 1}.
  \]
In particular, if $g\ge 10\sqrt{\lambda\Delta}\cdot\log(\lambda\Delta)$, then $\lambda_{\max}(\Inf_\mu) =O(\lambda+1)$.
\end{theorem}

\Cref{thm:matching-SI-large-girth} is proved in \Cref{sec:matching-SI}.
The spectral independence result in \Cref{thm:matching-tree} follows as a special case of \Cref{thm:matching-SI-large-girth} when $g=\infty$.
The spectral gap in \Cref{thm:matching-tree} is proved in \Cref{sec:miss-proof}  by applying existing techniques.

The trade-off between girth and spectral independence observed in \Cref{thm:matching-SI-large-girth} may be somewhat intrinsic to the monomer-dimer model.
This is showcased by the following extreme example with girth equals 2 (i.e.~the graph $G$ is a multigraph with parallel edges).
\begin{example} \label{thm:matching-SI-lb}
  %
Let $C_n$ be a cycle consisting of $n$ vertices for a sufficiently large $n$.  
Let $G$ be a multigraph obtained by replacing each edge of $C_n$ with $\Delta/2$ parallel edges, where $\Delta \geq 100$.
The Gibbs distribution  $\mu$ of the monomer-dimer model on $G$ with fugacity $\lambda = 1$ has    $\lambda_{\max}(\Psi_{\mu}) \ge \frac{\sqrt{\Delta}}{10}$.

\end{example}
\Cref{thm:matching-SI-lb} is proved in \Cref{sec:matching-SI-lb}.
This counterexample shows that the conjectured constant spectral independence of the monomer-dimer model in \Cref{conjecture:monomer-dimer-SI} requires a minimum girth of 3 to prevent the presence of parallel edges.

We consider the girth-2 scenario to be an exceptional case. 
Accordingly, we refine above \Cref{conjecture:monomer-dimer-SI} to apply specifically to simple graphs.


\begin{conjecture} \label{con:mathcing-O(1)-SI}
For any \emph{simple} graph $G=(V,E)$ and any $\lambda>0$,
the Gibbs distribution $\mu$ of the monomer-dimer model on $G$ with fugacity $\lambda$ has 
$\lambda_{\max}(\Psi_\mu)=O_{\lambda}(1)$.
\end{conjecture}

\Cref{thm:matching-SI-large-girth} shows that the conjecture is true when the girth $g$ is as large as some $g=\Omega(\sqrt{\lambda\Delta}\cdot\log(\lambda\Delta))$, and \Cref{thm:matching-SI-lb} shows that the conjecture is false when the girth $g=2$. The remaining unresolved cases are graphs with girth $3\le   g \leq O(\sqrt{\lambda\Delta}\cdot\log(\lambda\Delta))$.


\paragraph{The hardcore model.}
For the hardcore model, a compelling computational phase transition occurs at the uniqueness threshold $\lambda_c(\Delta)=\frac{(\Delta-1)^{\Delta-1}}{(\Delta-2)^\Delta}\approx \frac{\e}{\Delta}$.
When $\lambda<\lambda_c(\Delta)$, the Gibbs distribution has bounded total influence and the Glauber dynamics has optimal mixing time~\cite{weitz2006counting,anari2020spectral,chen2020rapid,chen2021optimal,chen2021rapid,anari2022entropic,chen2022optimal,chen2022localization}. 
And when $\lambda>\lambda_c(\Delta)$,  sampling becomes computationally intractable~\cite{sly2010computational,sly2012computational,galanisSV16}.

We then focus on the hardcore model on trees, where sampling remains tractable even beyond the uniqueness threshold.
Pioneered by a series of seminal works~\cite{martinelli2003ising, martinelli2004fast,Berger2005Glauber},
the mixing time in this case has been proved to be always polynomially bounded.
On one hand, when $\lambda < \lambda_c(\Delta)$, the Glauber dynamics is known to achieve optimal spectral gap and mixing time~\cite{chen2021rapid,chen2022optimal};
on the other hand, for sufficiently large $\lambda > C$, a polynomially large mixing lower bound has been proved~\cite{restrepo2014phase}, where $C \approx 28$ as pointed out by~\cite{efthymiou2023optimal}.
A critical problem for the hardcore model on trees is then to determine the critical threshold for fast mixing with optimal spectral gap.
Given the universality of the spectral independence~\cite{anari2023universality}, settling this problem is inherently  related to the spectral independence of the hardcore model on trees.

A recent work~\cite{efthymiou2023optimal} shows that when $\lambda < 1.3$, the spectral independence for the hardcore model on trees remains bounded, even though the total influence on trees may become unbounded beyond uniqueness~\cite{anari2020spectral, chen2020rapid}.
The following conjecture is proposed based on the reconstruction threshold.

\begin{conjecture}[\cite{efthymiou2023optimal}]\label{conjecture:hardcore-SI}
For any  tree $T$ with $n$ vertices, 
the spectral gap of the Glauber dynamics for the hardcore model on $T$ with fugacity $\lambda<\e-1$ is at least  $\Omega(n^{-1})$.
\end{conjecture}

Using our approach of approximate inverse, we prove \Cref{conjecture:hardcore-SI} and further push the threshold of $\lambda$ for spectral independence of hardcore model on trees to $\lambda<\e^2$.

\begin{theorem} \label{thm:hardcore-tree}
  Let $T=(V,E)$ be a tree of $n$ vertices, and $0 < \lambda <(1-\delta) \e^2$ for some $\delta \in \tp{ 0, 1/10 }$.
  The Gibbs distribution $\mu$ of the hardcore model on $T$ with fugacity $\lambda$ has the  spectral independence 
  \[
  \lambda_{\max}(\Inf_\mu) \leq \frac{36}{\delta^2}.
  \]
  Moreover, the Glauber dynamics on $\mu$ has asymptotically optimal spectral gap $\Omega_\delta(n^{-1})$.
  
  If further the maximum degree $\Delta$ of $T$ is bounded, the mixing time of the Glauber dynamics is bounded by $O_{\Delta,\delta}(n \log n)$.
\end{theorem}
It is widely believed that the critical threshold for fast mixing on trees aligns with the reconstruction threshold~\cite{efthymiou2023optimal},
which is known to be lower bounded as $\lambda_r (\Delta) > \e-1$~\cite{martin2003reconstruction}, leading to \Cref{conjecture:hardcore-SI}.
Our findings suggest that either fast mixing holds beyond reconstruction, or the reconstruction threshold is actually higher as $\lambda_r (\Delta)\ge \e^2$.

The spectral independence in \Cref{thm:hardcore-tree} is proved in \Cref{sec:hardcore-SI}.
The spectral gap and mixing time bounds in \Cref{thm:hardcore-tree} are proved using existing techniques in \Cref{sec:miss-proof}.

As noted in~\cite{efthymiou2023optimal}, there exist trees where the Glauber dynamics exhibits a polynomial slowdown when $\lambda > C$ for some constant $C \approx 28$. 
Combining this with the findings of~\cite{anari2023universality}, it is established that the spectral independence is unbounded for the hardcore model on trees when $\lambda$ is sufficiently large. 
Nonetheless, for completeness, we provide a lower bound on the unboundedness of the spectral independence for the hardcore model on trees when $\lambda$ is sufficiently large.

\begin{theorem}\label{thm:unbounded-hardcore}
  Let $\lambda > 28.15$. For any finite $C>0$, there exists a tree $T=(V,E)$ such that the Gibbs distribution $\mu$ of the hardcore model on tree $T$ with fugacity $\lambda$ satisfying that $\lambda_{\max}(\Inf_\mu) \ge C$.
\end{theorem}
\Cref{thm:unbounded-hardcore} is proved  in \Cref{sec:unbounded-hardcore}.

\subsection{Related work}
\paragraph{High dimensional expander.}
The framework of high-dimensional expander, introduced and significantly developed in recent years~\cite{kaufman2017high,dinur2017high,oppenheim2018local,kaufman2020high}, has been extended to Gibbs sampling in the seminal works~\cite{alev2020improved,anari2020spectral}. 
Subsequent research has further extended this framework from various perspectives~\cite{chen2020rapid,chen2021optimal,chen2021rapid,chen2022optimal,chen2022localization,anari2022entropic,chen2023nearlinear,chen2023matching,anari2023parallel,anari2022optimal,anari2023universality,chen2023strong}.

\paragraph{(Matrix) trickle-down method.}
Previously, a primary method for upper bounding the spectral independence $\lambda_{\max}(\Psi_\mu)$ without relying on bounding the total influence $\norm{\Psi_\mu}_\infty$ was the (matrix) trickle-down method.
This technique has been successfully applied to establish spectral independence results for the matroid bases~\cite{oppenheim2018local,anari2019logconcave} and edge coloring~\cite{abdolazimi2021matrix,wang2023sampling}.

\paragraph{Lifting spectral independence from tree to graph.}
\Cref{thm:matching-SI-large-girth}  basically lifts the spectral independence bound from trees to general graphs using the approximate inverse of the influence matrix.
Previously, for total influence $\norm{\Psi_\mu}_\infty$, such a lifting from trees to general graphs, transforming extremal total influence bounds from  trees to general graphs, has been achieved for $2$-spin systems through the use of  self-avoiding walk trees~\cite{godsil1993algebratic,weitz2006counting,bayati2007simple,li2013correlation,anari2020spectral,chen2020rapid,chen2021optimal,chen2023uniqueness}.
However, self-avoiding walk trees are restricted to 2-spin models.
Very recently, \cite{chen2023strong} introduced a recursive coupling technique to lift total influence bounds from trees to graphs with large girths, applicable to multi-spin models such as proper colorings.


\paragraph{Glauber dynamics on tree.}
The mixing time and the spectral gap of Glauber dynamics for Gibbs distribution on trees have been extensively studied~\cite{martinelli2003ising,Berger2005Glauber,martinelli2004fast,lucier2009glauber,goldberg2010mixing,tetali2012phase,restrepo2014phase,sly2017glauber,delcourt2020glauber,eppstein2023rapid,chen2023combinatorial}.
 In contrast to general graphs, the rapid mixing of Glauber dynamics on trees is believed to be closely related to the \emph{reconstruction threshold}, a topic of significant interest.
 
For the hardcore model on trees,  a series of research works has established both  upper  and lower bounds for the mixing time of Glauber dynamics~\cite{martinelli2003ising,martinelli2004fast,Berger2005Glauber,restrepo2014phase,eppstein2023rapid,chen2023combinatorial}.
For {complete $b$-ary trees}, the mixing time of Glauber dynamics is always $O(n \log n)$~\cite{martinelli2004fast}.
However, for arbitrary trees,  the mixing time shows a critical behavior. 
In particular, when $\lambda$ is sufficiently small, the mixing time of Glauber dynamics was proved to be $O(n \log n)$~\cite{martinelli2003ising,efthymiou2023optimal}; in contrast, when $\lambda$ is sufficiently large, the mixing time is lower bounded by a large polynomial~\cite{restrepo2014phase}. 
The critical threshold is believed to coincide with the reconstruction threshold $\lambda_r(\Delta)$~\cite{bhatnagar2010}. In~\cite{martin2003reconstruction}, it is proved that $\lambda_r (\Delta) > \e-1$, which leads to a conjecture that the mixing time of Glauber dynamics on trees  is $O(n \log n)$ when $\lambda < \e-1$~\cite{efthymiou2023optimal}. 

\section{Preliminaries}

\subsection{Notations and conventions}

Let $\mu$ be a distribution over $2^{[n]}$, where $[n] = \set{1, \ldots, n}$ is a set of $n$ elements.
Throughout this paper, 
we refer $\Pr[\mu]{i}$ and $\Pr[\mu]{\ol{i}}$ to the marginal probability that $i \in [n]$ is occupied and unoccupied respectively, i.e. $\Pr[S\sim \mu]{i\in S}$ and $\Pr[S \sim \mu]{i\not\in S}$.
Moreover, denote the probability of $j$ being occupied conditioned on $i$ being occupied and unoccupied by $\Pr[\mu]{j \mid i}$ and $\Pr[\mu]{j \mid \ol{i}}$ respectively.
When the context is clear, we may drop the subscript $\mu$.
When the formula is heavy, we may also use $\mu_i, \mu_{\ol{i}}, \mu_j^i, \mu_j^{\ol{i}}$ to replace $\Pr{i}$, $\Pr{\ol{i}}$, $\Pr{j\mid i}$, $\Pr{j\mid \ol{i}}$, respectively, to make the expression clean.


For a graph $G = (V,E)$, we may write $N(v)$ to represent the set of vertices $u$ adjacent to $v$.
We write $e \sim u$ or $u \sim e$ if vertex $u \in V$ is incident to edge $e \in E$, and $E_u$ is the set of edges including $u$.
Furthermore, we write $e \overset{u}{\sim} f$ for distinct edges $e,f \in E$ if vertex $u$ is incident to both $e$ and $f$.

For some positive parameter $p$, we will use $O_p(\cdot)$ and $\Omega_p(\cdot)$ to hide the factor related to $p$, respectively.
That is, $O_p(\cdot) = f(p) O(\cdot)$ and $\Omega_p(\cdot) = g(p)\Omega(\cdot)$ for some positive function $f$ and $g$.



\subsection{Markov chains and mixing time}
Let $(X_t)_{t \ge 0}$ be a (discrete-time) Markov chain over a finite state space $\Omega$ with the transition matrix $P \in \mathbb{R}^{\Omega \times \Omega}$.
We may refer to the Markov chain by its transition matrix $P$.
The Markov chain is
\begin{enumerate}
  \item \emph{irreducible}, if for any state $x,y \in \Omega$, there exists $t \ge 1$ satisfying $P^t(x,y)>0$;
  \item \emph{aperiodic}, if for any state $x \in \Omega$, $\gcd\set{t \ge 0: P^t(x,x) > 0} = 1$.
\end{enumerate}
If the Markov chain $P$ is irreducible and aperiodic, the Markov chain has a unique \emph{stationary distribution} $\mu$, i.e., $\mu P = P$.
Moreover, we say a distribution $\mu$ is \emph{reversible} with respect to Markov chain $P$, if the stationary distribution satisfies the \emph{detailed balance condition}, i.e., $\mu(x) P(x,y) = \mu(y) P(y,x)$ for all $x,y \in \Omega$.
It is known that $\mu$ is the stationary distribution of the Markov chain $P$ if $\mu$ is reversible to $P$.

The \emph{mixing time} of Markov chain $P$ measures the speed of convergence towards the stationary distribution. Formally, 
\begin{align*}
  T_{\mathrm{mix}} = \max_{x \in \Omega}\min \set{t \ge 0 : \DTV{P^t(x,\cdot)}{\mu} < \frac{1}{4}},
\end{align*}
where $\DTV{\mu}{\nu} = \frac{1}{2} \sum_{x \in \Omega} \abs{\mu(x) - \nu(x)}$ is the total variation distance.

A canonical single-site Markov chain for sampling high dimensional distribution is the \emph{Glauber dynamics}.
Let $\mu$ be a distribution over $2^{[n]}$. 
In each step, the Glauber dynamics updates a configuration $X \in 2^{[n]}$ according to the following steps:
\begin{enumerate}
  \item select $u \in [n]$ uniformly at random;
  \item 
    update $X$ to
    \begin{align*}
      \begin{cases}
       X \cup \set{u}, & \text{with probability } \frac{\mu(X\cup \set{u})}{\mu(X\cup \set{u}) + \mu(X\setminus \set{u})} \\
       X \setminus \set{u}, & \text{otherwise}.
      \end{cases}
    \end{align*}
  \end{enumerate} 
It is known that $\mu$ is reversible with respect to the Glauber dynamics.

\subsection{Spectral independence and related matrices}
Spectral independence defined below is a notion introduced in~\cite{anari2020spectral} that characterizes the correlation of variables in a high dimensional distribution $\mu$. 
\begin{definition}[influence matrix] \label{def:inf-mat}
  Let $\mu$ be a distribution over $2^{[n]}$. Its \emph{influence matrix} $\Inf_\mu \in \^R^{n \times n}$ is defined by
  \begin{align*}
    \forall i, j \in [n], \quad \Inf_\mu(i, j) =
    \begin{cases}
      \Pr[\mu]{j \mid i} - \Pr[\mu]{j \mid \overline{i}} & \text{ if } \Pr[\mu]{i} \in (0, 1),\\
      0 & \text{ otherwise.}
    \end{cases}
  \end{align*}
\end{definition}

\begin{definition}[spectral independence~\cite{anari2020spectral}] \label{def:SI}
  Let $\eta \geq 0$, a distribution $\mu$ over $2^{[n]}$ is said to be \emph{$\eta$-spectrally independent} if $\lambda_{\max}(\Inf_\mu) \leq \eta$.
\end{definition}

We introduce the following matrices that are closely related to the influence matrix.

\begin{definition}[correlation matrix~\cite{anari2022entropic}] \label{def:cor-mat}
  Let $\mu$ be a distribution over $2^{[n]}$. Its \emph{correlation matrix} $\Cor_\mu \in \^R^{n \times n}$ is defined by
  \begin{align*}
    \forall i, j \in [n], \quad \cor_\mu(i, j) =
    \begin{cases}
      \Pr[\mu]{j \mid i} - \Pr[\mu]{j} & \text{ if } \Pr[\mu]{i} \in (0, 1), \\
      0 & \text{ otherwise.}
    \end{cases}
  \end{align*}
\end{definition}

\begin{definition}[symmetrized influence matrix] \label{def:sym-inf-mat}
  Let $\mu$ be a distribution over $2^{[n]}$. Its \emph{symmetrized influence matrix} $\Inf_\mu \in \^R^{n \times n}$ is defined by
  \begin{align*}
    \forall i, j \in [n], \quad \Cor_\mu(i, j) =
    \begin{cases}
      \frac{\Pr[\mu]{j, i} - \Pr[\mu]{j}\Pr[\mu]{i}}{\sqrt{\Pr[\mu]{j}\Pr[\mu]{\overline{j}}}\sqrt{\Pr[\mu]{i}\Pr[\mu]{\overline{i}}}} & \text{ if } \Pr[\mu]{j}, \Pr[\mu]{i} \in (0, 1),\\
      0 & \text{ otherwise.}
    \end{cases}
  \end{align*}
\end{definition}
\begin{remark}
  We remark that $\Cor$ is usually regarded as the correlation matrix of random variables.
  Here, we use the name ``symmetrized influence matrix'' to avoid confusion with the correlation matrix used in previous works~\cite{anari2022entropic}.
\end{remark}

\begin{remark}
  We remark that $\Cor_\mu(i,j) = \sqrt{\Psi_\mu(i,j) \Psi_\mu(j,i)}$ for any $i,j \in [n]$.
\end{remark}

The following proposition shows that the influence matrix, correlation matrix, and the symmetrized influence matrix could be transformed to each other by multiplying diagonal matrices. 

\begin{proposition} \label{prop:mat-rel}
  Let $\mu$ be a distribution over $2^{[n]}$.
  Let $\Pi := \overline{D}D$, where $D := \-{diag}\set{\Pr[\mu]{i}}_{i\in [n]}$, $\overline{D} := \-{diag}\set{\Pr[\mu]{\overline{i}}}_{i\in [n]}$ are diagonal matrices, then it holds that
  \begin{align}
    \label{eq:sym-inf}
    \Cor_\mu &= \Pi^{1/2}\Inf_\mu\Pi^{-1/2}  \\
    \label{eq:cor-inf}
    \text{and} \quad \cor_\mu &= \overline{D}\Inf_\mu.
  \end{align}
\end{proposition}

\begin{proof}
  We first prove \eqref{eq:sym-inf}.
  By the Law of Conditional Probability, it holds that
  \begin{align*}
    \Pr[\mu]{j,i}-\Pr[\mu]{j}\Pr[\mu]{i}
    &=\Pr[\mu]{i}\left(\Pr[\mu]{j\mid i}-\Pr[\mu]{j,i}-\Pr[\mu]{j,\overline{i}}\right) \\
    &=\Pr[\mu]{i}\left(\Pr[\mu]{j\mid i}\left(1-\Pr[\mu]{i}\right)-\Pr[\mu]{j,\overline{i}}\right)\\
    &=\Pr[\mu]{i}\Pr[\mu]{\overline{i}}\left(\Pr[\mu]{j \mid i} - \Pr[\mu]{j \mid \overline{i}}\right).
  \end{align*}
  Therefore, for any $i,j\in[n]$, it holds that
  \begin{equation*}
    \Pi^{1/2}\Inf_\mu\Pi^{-1/2}(i,j)=\sqrt{\frac{\Pr[\mu]{i}\Pr[\mu]{\overline{i}}}{\Pr[\mu]{j}\Pr[\mu]{\overline{j}}}}\frac{\Pr[\mu]{j,i}-\Pr[\mu]{j}\Pr[\mu]{i}}{\Pr[\mu]{i}\Pr[\mu]{\overline{i}}}=\Cor_\mu(i, j). 
  \end{equation*}

  Then, we prove \eqref{eq:cor-inf}.
  For any $i,j\in[n]$, it holds that
  \begin{align*}
    \overline{D}\Inf_\mu(i,j)
    &=\Pr[\mu]{\overline{i}}\left(\Pr[\mu]{j\mid i}-\Pr[\mu]{j\mid \overline{i}}\right)\\
    &=\Pr[\mu]{j\mid i}-\Pr[\mu]{i}\Pr[\mu]{j\mid i}-\Pr[\mu]{\overline{i}}\Pr[\mu]{j\mid \overline{i}}\\
    &=\Pr[\mu]{j\mid i}-\Pr[\mu]{j, i}-\Pr[\mu]{j, \overline{i}}\\
    &=\Pr[\mu]{j\mid i}-\Pr[\mu]{j}= \cor_\mu(i,j). \qedhere
  \end{align*}

\end{proof}

According to~\Cref{prop:mat-rel}, the maximum eigenvalue of $\Inf$ and $\cor$ can be easily related to each other when there is some guarantee on the marginals.

\begin{corollary} \label{cor:inf-cor-lambda-max}
    If for every $i \in [n]$, $\Pr{\overline{i}} \in [l, r] \subseteq (0, 1)$.
    Then, it holds that
    \[\lambda_{\-{max}}(\cor_\mu) r^{-1}
        \leq \lambda_{\-{max}}(\Inf_\mu)
        \leq \lambda_{\-{max}}(\cor_\mu) l^{-1}.\]
\end{corollary}
\begin{proof}
    According to \Cref{prop:mat-rel}, it holds that $\Inf_\mu$ is a self-adjoint operator with respect to the inner product $\inner{\cdot}{D\overline{D} \cdot}$, and $\cor_\mu$ is a self-adjoint operator with respect to the inner product $\inner{\cdot}{D \cdot}$.
    By the Courant-Fischer theorem,
    \begin{align*}
        \lambda_{\-{max}}(\Inf_\mu) 
        &= \max_{f\neq 0} \frac{\inner{f}{D\overline{D} \Inf_\mu f}}{\inner{f}{D\overline{D} f}}
        \overset{\eqref{eq:cor-inf}}{=} \max_{f \neq 0} \frac{\inner{f}{D \cor_\mu f}}{\inner{f}{D f}} \cdot \frac{\inner{f}{D f}}{\inner{f}{D\overline{D} f}}.
    \end{align*}
    Hence, it holds that
    \begin{align*}
        \max_{f\neq 0} \frac{\inner{f}{D \cor_\mu f}}{\inner{f}{D f}} \cdot \min_{f\neq 0} \frac{\inner{f}{D f}}{\inner{f}{D\overline{D} f}} 
        &\leq \lambda_{\-{max}}(\Inf_\mu)
        \leq \max_{f \neq 0} \frac{\inner{f}{D \cor_\mu f}}{\inner{f}{D f}} \cdot \max_{f \neq 0} \frac{\inner{f}{D f}}{\inner{f}{D\overline{D} f}} \\
        \Longrightarrow \lambda_{\-{max}}(\cor_\mu) r^{-1}
        &\leq \lambda_{\-{max}}(\Inf_\mu)
        \leq \lambda_{\-{max}}(\cor_\mu) l^{-1}. \qedhere
    \end{align*}
\end{proof}

\subsection{Spectral independence of $k$-transformed distributions}
The notion of $k$-transformation was initiated in~\cite{chen2021rapid} to analyze the spectral gap of the field dynamics.
It has been studied extensively during the research on the optimal mixing time of the Glauber dynamics on two-state spin systems for unbounded degree graph~\cite{chen2021rapid,anari2021entropicII,chen2022optimal}.

\begin{definition}[$k$-transformation] \label{def:k-trans}
  Let $\mu$ be a distribution over $2^{[n]}$ and $k \geq 1$ be an integer.
  The \emph{$k$-transformation} (a.k.a. \emph{$k$-blow up}) of $\mu$, denoted by $\mu_k$, is a distribution over $2^{[n]\times [k]}$.

  Let $X \sim \mu$, then $\mu_k$ is the distribution of $Y \in 2^{[n]\times [k]}$ constructed as follows:
  for every $i \in X$, add $(i, t)$ to $Y$ where $t$ is picked from $[k]$ uniformly at random.
\end{definition}

Under the $k$-transformation, the Gibbs distribution $\mu$ of the monomer-dimer model specified by graph $G$ and fugacity $\lambda$
will transform to the Gibbs distribution of the monomer-dimer model specified by graph $G_k$ and fugacity $\lambda/k$, 
where $G_k$ is the graph obtained by replacing all edges in $G$ with $k$ parallel edges.

Similarly, the Gibbs distribution $\mu$ of the hardcore model specified by graph $G$ and fugacity $\lambda$
will transform to the Gibbs distribution of the hardcore model specified by graph $G_k$ and fugacity $\lambda/k$, 
where $G_k$ is the graph obtained by replacing all vertices in $G$ with a $k$-clique and fully connecting the adjacent $k$-cliques.

In the original analysis of the field dynamics, an important step is to show that the notion of spectral independence could be approximately preserved along the $k$-transformation.
In~\cite{chen2021rapid}, an approximate relationship between $\mu$ and $\mu_k$ is revealed ($\lambda_{\max}(\abs{\Inf_{\mu_k}}) \leq \lambda_{\max}(\abs{\Inf_\mu}) + 2$).
Then, \cite{anari2021entropicII} refined this approximate relationship and got an exact relationship between $\mu$ and $\mu_k$ in terms of $\cor$.

\begin{proposition}[\text{\cite[Proposition 26]{anari2021entropicII}}] \label{prop:k-trans-cor}
  Let $\eta \geq 1$, $k \in \^N$, $\mu$ is a distribution over $2^{[n]}$ and $\mu_k$ over $2^{[n]\times [k]}$ is its $k$-transformation.
  If $\lambda_{\max}(\cor_\mu) = \eta$, then it holds that $\lambda_{\max}(\cor_{\mu_k}) = \eta$.
\end{proposition}

We note that \Cref{prop:k-trans-cor} is a special case of its original version in~\cite{anari2021entropicII}.

\subsection{Inverse of product distance matrix}
\label{sec:product-distance-matrix}

A \emph{product distance matrix} on graph $G=(V,E)$ is a matrix $P\in \mathbb{R}^{V\times V}$ satisfying the following properties:
\begin{enumerate}
  \item $P(u,u)=1$, for any vertex $v\in V$.
  \item If $u,v\in V$ are vertices such that each directed path from $u$ to $v$ passed through the cut-vertex $w$,
          then $P(u,v)=P(u,w)\cdot P(w,v)$.
\end{enumerate}
A \emph{block} of graph $G=(V,E)$ is a maximal subgraph without cut-vertex.

Observed in~\cite{bapat2013product}, the inverse of a product distance matrix is well-structured.

\begin{proposition}[\cite{bapat2013product}]\label{prop:product-inverse}
  Let $P \in \mathbb{R}^{V \times V}$ be a product distance matrix on graph $G=(V,E)$, and $\+B$ be the set of blocks in graph $G$. 
  For each block $B \in \+B$, let $P_B$ be the principal matrix of $P$ on vertices in $B$.
  The inverse $P^{-1}$ of $P$ satisfies
  \begin{align*}
    \forall u,v \in V, \quad 
    P^{-1}(u,v) =
    \begin{cases}
      1 + \sum_{\substack{B \in \+B\\u \in B}} (P_B^{-1}(u,u) - 1)& u = v,\\
      P_B^{-1}(u,v) & \{u,v\} \subseteq B,\\
      0 & \text{otherwise.}
    \end{cases} 
  \end{align*} 
\end{proposition}

\begin{remark}
  We may also call a matrix $P \in \^R^{E\times E}$ a product distance matrix if it is a product distance matrix on the line graph of $G$.
\end{remark}

\subsection{Linear algebra toolkit}

To apply \Cref{prop:product-inverse}, we need to calculate the inversion of the (symmetrized) influence matrix on each block of the graph.
In practice, we will use the Sherman-Morrison-Woodbury formula.

\begin{proposition}[Sherman-Morrison-Woodbury formula] \label{prop:inversion-block}
  Let $\Lambda\in \mathbb{R}^{n\times n}$ be an invertible square matrix and $x,y\in \mathbb{R}^n$ be column vectors.
  If $\*y^\top \Lambda^{-1}\*x + 1 \neq 0$, then
  \[(\Lambda + \*x\*y^\top)^{-1} = \Lambda^{-1} - \frac{\Lambda^{-1}\*x\*y^\top\Lambda^{-1}}{1 + \*y^\top\Lambda^{-1}\*x}.\]
\end{proposition}

Moreover, when the matrix has the form in~\Cref{prop:inversion-block}, its determinant can be easily calculated via the matrix determinant lemma. 
\begin{proposition}[matrix determinant lemma]\label{prop:det-lem}
  Let $\Lambda\in \mathbb{R}^{n\times n}$ be an invertible square matrix and $x,y\in \mathbb{R}^n$ be column vectors, then
  \begin{equation*}
    \det(\Lambda + \*x\*y^\top) =\det(\Lambda) \left(1 + \*y^\top\Lambda^{-1}\*x\right).
  \end{equation*}
\end{proposition}

In the context of matrix $A$ having real eigenvalues, it is customary to denote them as $\lambda_1(A)\ge\cdots\ge\lambda_n(A)$.
The interlacing theorem states that for a real symmetric matrix of order $n$, 
the eigenvalues of the matrix interlace with the eigenvalues of the principal submatrix of order $n-1$. 

\begin{proposition}[interlacing theorem]\label{prop:interlace}
  Let $A\in \^R^{n \times n}$ be a symmetric matrix and $B$ be a $(n-1)\times (n-1)$ principal submatrix of $A$.
  Then, it holds that
  \begin{equation*}
    \lambda_1(A)\ge\lambda_1(B)\ge\lambda_2(A)\ge\lambda_2(B)\ge\cdots\ge\lambda_{n-1}(A)\ge\lambda_{n-1}(B)\ge\lambda_{n}(A).
  \end{equation*}
\end{proposition}

When a symmetric matrix is perturbed by another symmetric matrix, we can use Weyl's inequality to bound the changes to its eigenvalues.
\begin{proposition}[Weyl's inequality]\label{prop:weyl}
  Let $A,B\in \^R^{n \times n}$ be two symmetric matrices, then for any $i,j\in[n]$, it holds that
  \begin{equation*}
    \lambda_{i+j-1}(A+B)\le \lambda_i(A)+\lambda_j(B)\le\lambda_{i+j-n}(A+B).
  \end{equation*}
\end{proposition}

\section{Method of approximate inverse}\label{sec:proof-outline}

Let $G=(V,E)$ be an undirected graph and $\mu$ be a distribution over $2^E$ or $2^V$.
If $\mu$ is supported over $2^E$, it is called an \emph{edge model distribution}; otherwise, it is called a \emph{vertex model distribution}.

Throughout this section, we will assume that all distributions satisfy the conditional independence property, formally defined below.
\begin{assumption}[conditional independence] \label{def:conf-ind}
  Suppose that $\mu$ is a vertex (or edge) model distribution. Let $u, v, w \in V\; (\text{or }E)$. If the removal of $v$ will disconnect $u$ and $w$ in $G$, then 
  \begin{align*}
    \text{for } v^* \in \set{v, \ol{v}}, w^* \in \set{w, \ol{w}}, \quad
    \Pr[\mu]{u \mid v^*, w^*}  = \Pr[\mu]{u \mid v^*}. 
  \end{align*}
\end{assumption}
\begin{remark}
  In practice, the spectral independence is usually asked for all pinnings of $\mu$. Hence, the conditional independence property is usually assumed for all pinnings of $\mu$.
 However, in this section, we only need to assume the conditional independence for $\mu$.
\end{remark}


Our method relies on the following \emph{approximate inverse} to analyze the spectral independence of $\mu$, specifically the spectrum of $\Cor$.
This form of the approximate inverse is inspired by~\cite{bapat2013product}, where the inverse of $\Cor$ is well-structured when the underlying graph is a tree.

\begin{definition}[approximate inverse for edge model] \label{def:Q-edge}
  The approximate inverse for $\Cor$ of an edge model distribution $\mu$ is constructed as follows:
  \begin{enumerate}
    \item For each $v \in V$, let $\Cor_v \in \^R^{E_v \times E_v}$ be the principal minor of $\Cor$ on $E_v$, the edges incident to $v$.
    \item Let $\widehat{Q}_v := (\Cor_v)^{-1}$ and let $Q_v \in \^R^{E\times E}$ be the matrix obtained by augmenting $\widehat{Q}_v$ with zero entries.
    \item Finally, let $Q := \sum_{v\in V} Q_v - I$.
  \end{enumerate} 
\end{definition}

\begin{definition}[approximate inverse for vertex model] \label{def:Q-vertex}
  The approximate inverse for $\Cor$ of a vertex model distribution $\mu$  is constructed as follows:
  \begin{enumerate}
    \item For each $e = \{u,v\} \in E$, let $\Cor_e \in \^R^{e \times e}$ be the principal minor of $\Cor$ on $e = \{u,v\}$.
    \item Let $\widehat{Q}_e := (\Cor_e)^{-1}$ and let $Q_e \in \^R^{V \times V}$ be the matrix obtained by augmenting $\widehat{Q}_e$ with zero entries accordingly.
    \item Finally, let $Q := \sum_{e \in E} Q_e - \diag \set{d_u-1}_{u \in V}$, where $d_u$ is the degree of vertex $u$.
  \end{enumerate}
\end{definition}

As observed in~\cite{anari2020spectral}, according to the conditional independence property, when $G$ is a tree, both $\Psi$ and $\Cor$ are product distance matrices (see~\Cref{sec:product-distance-matrix} for the detailed definition).
We formalize this into the following proposition.
\begin{proposition}[\text{\cite[Lemma B.2]{anari2020spectral}}] \label{prop:product-distance}
  Under \Cref{def:conf-ind}, both $\Psi$ and $\Cor$ are product distance matrices.
\end{proposition}
By~\Cref{prop:product-inverse}, we can derive the following proposition.
\begin{proposition}\label{prop:Cor-Q-I}
  When the graph $G = (V,E)$ is a tree, it holds that $Q\; \Cor = I$.
\end{proposition}

When $G$ is not a tree, the matrix $Q$ may serve as an approximate inverse of $\Cor$, 
which is generally easier to analyze and can be a good estimator of $\tp{\Cor}^{-1}$ in certain circumstances. 
This observation is formalized as the following conditions.

\begin{condition}[$\alpha$-boundedness]\label{cond:boundedness}
  Let $Q$ be the approximate inverse of $\Cor$. The minimum eigenvalue of $Q$ is lower bounded by $\alpha$ for some $\alpha > 0$.
\end{condition}

\begin{condition}[$\beta$-approximate ratio]\label{cond:approximate}
  Let $Q$ be the approximate inverse of $\Cor$. The maximum eigenvalue of $Q \Cor$ is upper bounded by $\beta$ for some $\beta > 0$.
\end{condition}

If these conditions can be verified, the maximum eigenvalue of $\Cor$ and $\Psi$ can be upper bounded as follows.
\begin{theorem} \label{thm:Psi-lambda-max}
  For a vertex (or edge) model satisfying both \Cref{cond:boundedness} and~\ref{cond:approximate}, it holds that
  \begin{align*}
    \lambda_{\max}(\Psi) = \lambda_{\max}(\Cor) \le \frac{\beta}{\alpha}.
  \end{align*}
  Moreover, if the underlying graph $G$ is a tree,~\Cref{cond:approximate} holds with $\beta = 1$.
\end{theorem}

\begin{proof}
 First, note that
  \begin{align*}
    \lambda_{\max}(\Cor) &= \lambda_{\max} \tp{Q^{-1/2} Q^{1/2} \Cor Q^{1/2} Q^{-1/2}}\\
    &= \max_{\*x \neq 0} \frac{\*x^\top Q^{-1/2} Q^{1/2} \Cor Q^{1/2} Q^{-1/2} \*x }{\*x^\top \*x}\\
    &= \max_{\*x \neq 0} \frac{\*x^\top Q^{-1} \*x}{\*x^\top \*x} \cdot \frac{\tp{\*x^\top Q^{-1/2} } Q^{1/2} \Cor Q^{1/2} \tp{Q^{-1/2} \*x}}{\*x^\top Q^{-1} \*x}\\
    &\le \lambda_{\max}(Q^{-1}) \lambda_{\max}(Q^{1/2} \Cor Q^{1/2})\\
    &= \lambda_{\max}(Q^{-1}) \lambda_{\max}(Q \; \Cor).
  \end{align*}
  Therefore, $\lambda_{\max}(\Psi) = \lambda_{\max} \tp{\Cor} \le \frac{\beta}{\alpha}$.

  The rest follows from~\Cref{prop:product-inverse} and~\Cref{prop:Cor-Q-I}.
\end{proof}

Although the approximate inverse simplifies the analysis, 
verifying \Cref{cond:boundedness} and \Cref{cond:approximate} may still be challenging. Next, we introduce several tools to help establish these conditions, respectively in edge and vertex models.

\subsection{Edge model}
Let $\mu$ be an edge model  distribution with underlying graph  $G=(V,E)$. 
We observe that the ``local spectral independence'' implies the boundedness of approximate inverse.

\begin{lemma}\label{lem:edge-boundedness}
  Let $\beta < 2$ be a constant.
  If $\Cor_u$, 
  the principal minor of $\Cor$ on the set of edges incident to $u$, 
  satisfies $\lambda_{\max}\tp{\Cor_u} \le \beta$ for all $u \in V$, 
  then~\Cref{cond:boundedness} holds with $\alpha = 2/\beta-1$.
\end{lemma}

\begin{proof}
  Since $\lambda_{\max}\tp{\Cor_u} \le \beta$ for all $u \in V$, we have $\tp{\Cor}^{-1} \succcurlyeq \frac{1}{\beta} I $. 
  Hence, 
  \begin{align*}
    Q-\tp{\frac{2}{\beta}-1}I = \sum_{v \in V} Q_v - \frac{2}{\beta} I = \sum_{v \in V} \tp{Q_v - \frac{1}{\beta} I_{E_v}} \succcurlyeq 0,
  \end{align*}
  where $I_{E_v} = \diag \set{\*1[e \in E_v]}_{e \in E}$, and the last inequality follows from the definition of $Q_v$ in~\Cref{def:Q-edge}.
\end{proof}

\begin{remark}
  For the monomer-dimer model, it can be verified that $\Cor_u \preccurlyeq \ol{D}^{-1}$ where $\ol{D} = \-{diag}\set{\mu_{\ol{e}}}_{e\in E_u}$.
  When $\lambda < 1$, we have $\mu_{\ol{e}} > 1/2$, which indicates that $\Cor_v \preccurlyeq \ol{D}^{-1} \prec 2 I$.
For a detailed analysis covering all $\lambda > 0$, please refer to \Cref{lem:lambda-min-Q}.
\end{remark}

Furthermore, we develop a systematic approach for lifting the spectral independence from trees to general graphs with large girth (i.e., ensuring a constant approximate ratio) when sub-constant total influence decay exhibits on trees, characterized as follows.

\begin{condition} \label{cond:SAW-tree}
  Let $G = (V, E)$ be a graph and let $\delta \in (0, 1)$ and $C > 0$ be two parameters.
  For every $u \in V$, there is a tree $T^u = (V^u, E^u)$ rooted at vertex $u$ and equipped with a function $\chi_u: E^u \to E$ and an influence matrix $\Psi^u \in \^R^{E^u \times E^u}$ of some distribution $\mu_T$ supported on $2^{E^u}$, together satisfying:
  \begin{enumerate}
  \item \label{item:struct} $\chi_u^{-1}(e) = \{e\}$ for all edges $e \in E_u$;
  \item \label{item:dist} $\-{dist}_{T^u}(u,g^\prime)\ge \-{dist}_{G_e}(u,g)$ for all edges $g \in E$ and $g^\prime\in \chi_u^{-1}(g)$;
  \item \label{item:influence} $\Inf(e, g) = \sum_{g' \in \chi_u^{-1}(g)} \Inf^u (e, g')$ for all edges $e \in E_u$ and $g \in E$;
  \item \label{item:DIT} for all $k\geq 0$ and $e \in E^u$, it holds that
    \begin{align*}
      \sum_{f:\-{dist}(e, f) = k} \abs{\Psi^u(e,f)} \leq C(1 - \delta)^k.
    \end{align*}
  \end{enumerate}
  Here, $\Psi$ represents the influence matrix for the Gibbs distribution on $G$.
\end{condition}

In above, \Cref{item:struct}, \Cref{item:dist} and \Cref{item:influence} in \Cref{cond:SAW-tree} are designed to relate the graphs and trees under the approximate inverse.
\Cref{item:DIT} in \Cref{cond:SAW-tree} was known by the name ``total influence decay'' in previous work~\cite{chen2023strong}, and it can be derived from the decay of correlation property~\cite{chen2020rapid,chen2021optimal}.
In practice, such tree is likely to be constructed through self-avoiding walks in $G$.
In particular, for the monomer-dimer model, it is worth noting that the self-avoiding walk tree (i.e., the path-tree) satisfies \Cref{cond:SAW-tree} with $C = 2$ and $\delta \approx 1/\sqrt{\lambda\Delta}$~\cite{godsil1993algebratic,chen2021optimal}. 
Application of \Cref{cond:SAW-tree} is deferred to  \Cref{sec:matching-SI}.

\begin{lemma} \label{lem:tree-graph}
If \Cref{cond:SAW-tree} is satisfied with parameter $\delta$ and $C$ by a graph  $G = (V, E)$ with girth $2g + 1$, then 
  \begin{align*}
    \lambda_{\max}(Q\Cor) \leq 2C \cdot (1 - \delta)^g/\delta + 1.
  \end{align*}
\end{lemma}

The rest of this section is dedicated to proving \Cref{lem:tree-graph}. 
The key idea behind the proof is the observation that
the approximate inverse effectively diminishes the influence of an edge $e$ on edge $f$ when $e$ is close to $f$, 
in graphs that is locally tree-like. 
This allows us to bound the row sum of $Q \; \Cor$ by leveraging the sub-constant decay rate of total influence,
provided that the graph has a sufficiently large girth.

The proof starts with a connection between the influence matrix $\Inf$ and its symmetric variant $\Cor$. 
We first introduce an analog of approximate inverse in~\Cref{def:Q-edge} for influence matrix $\Psi$. 

\begin{definition}[approximate inverse of $\Inf$] \label{def:W}
  The approximate inverse $W$ of $\Inf$ is constructed as follows:
  \begin{enumerate}
    \item For each $v \in V$, let $\Inf_v \in \^R^{E_v \times E_v}$ be the principal minor of $\Inf$ on $E_v$, the edges incident to $v$.
    \item Let $\widehat{W}_v := (\Inf_v)^{-1}$ and let $W_v \in \^R^{E\times E}$ be the matrix obtained by augmenting $\widehat{W}_v$ with zero entries accordingly.
    \item Finally, let $W := \sum_{v\in V} W_v - I$.
  \end{enumerate} 
\end{definition}

Note that $Q\;\Cor$ is similar to $W \Inf$. Therefore, the task of bounding $\lambda_{\max}(Q\;\Cor)$ is now reduced to bounding $\lambda_{\max}(W\Inf)$.

\begin{proposition} \label{prop:W-Q-similarity}
  It holds that $\lambda_{\max}(W \Inf) = \lambda_{\max}(Q\; \Cor)$.
\end{proposition}
\begin{proof}
  Recall that $\Pi := \-{diag}\set{\Pr{e}\Pr{\overline{e}}}_{e\in E}$.
  For a vertex $v \in V$, let $\Pi_v$ be the principal minor of $\Pi$ on $E_v$.
  According to \eqref{eq:sym-inf}, we have $\Cor_v = \Pi_v^{1/2}\Inf_v\Pi_v^{-1/2}$.
  Since $\widehat{Q}_v = (\Cor_v)^{-1}$ and $\widehat{W}_v = (\Inf_v)^{-1}$, this indicates that $\widehat{Q}_v = \Pi_v^{1/2}\widehat{W}_v\Pi_v^{-1/2}$.
  Hence, we have $Q_v = \Pi^{1/2}W_v\Pi^{-1/2}$ and $Q = \Pi^{1/2}W\Pi^{-1/2}$.
  Combining~\Cref{prop:mat-rel}, $Q\;\Cor = \Pi^{1/2}W\Inf\Pi^{-1/2}$.
\end{proof}

\begin{lemma} \label{lem:W-Psi-explicit}
  If \Cref{item:struct} and \Cref{item:influence} of \Cref{cond:SAW-tree} hold, then we have
  \begin{align*}
    \forall e = (u,v) \text{ and } f \in E, \quad W\Inf(e, f) =
      \sum_{\substack{f^u \in \chi_u^{-1}(f) \\ f^u \in T^u_v}}\Psi^u(e, f^u) + \sum_{\substack{f^v \in \chi_v^{-1}(f) \\ f^v \in T^v_u}} \Psi^v(e, f^v) - \Psi(e, f),
  \end{align*}
  where $T^u_v$ denotes the subtree of $T^u$ rooted at $v$, and $T^v_u$ is defined accordingly.
\end{lemma}

\Cref{lem:W-Psi-explicit} is proved  later in \Cref{sec:W-Psi-explicit}.
Intuitively, this proof is to transfer the quantitative relationship between $\widehat{W}_u$ and $\widehat{W}_v$ defined respectively on trees $T^u$ and $T^v$ (as in~\Cref{cond:SAW-tree}) to the graph $G$.
This follows from the fact that $\widehat{W}_u$ remains the same in both $G$ and $T^u$ by~\Cref{cond:SAW-tree}, and likewise for $\widehat{W}_v$.
%
%
Now, we are ready to prove \Cref{lem:tree-graph}.
\begin{proof}[Proof of \Cref{lem:tree-graph}]
  Suppose $e = (u, v)$, we define another graph $G_e := G - e$. By our assumption on girth,
  every edge $f \neq e$ in $E$ falls into exactly one of the following categories:
  \begin{enumerate}
  \item $E^u_\leq := \set{f \mid \-{dist}_{G_e}(f, u) \leq g}$;
  \item $E^v_\leq := \set{f \mid \-{dist}_{G_e}(f, v) \leq g}$;
  \item $E_> := \set{f \mid \-{dist}_{G_e}(f, u) > g \text{ and } \-{dist}_{G_e}(f, v) > g}$.
  \end{enumerate}
  For $f \in E^u_\leq$ or $f \in E_{>}$, by \Cref{lem:W-Psi-explicit}, we have
  \begin{align*}
    W\Inf(e, f)
    &= \sum_{\substack{f^u \in \chi_u^{-1}(f) \\ f^u \in T^u_v}}\Psi^u(e, f^u) + \sum_{\substack{f^v \in \chi_v^{-1}(f) \\ f^v \in T^v_u}} \Psi^v(e, f^v) - \Psi(e, f) \\
    (\text{\Cref{cond:SAW-tree}, \Cref{item:influence}})
    &= \sum_{\substack{f^u \in \chi_u^{-1}(f) \\ f^u \in T^u_v}}\Psi^u(e, f^u) + \sum_{\substack{f^v \in \chi_v^{-1}(f) \\ f^v \in T^v_u}} \Psi^v(e, f^v) - \sum_{f^v \in \chi_v^{-1}(f)}\Psi^v(e, f^v) \\
    &= \sum_{\substack{f^u \in \chi_u^{-1}(f) \\ f^u \in T^u_v}}\Psi^u(e, f^u) - \sum_{\substack{f^v \in \chi_v^{-1}(f) \\ f^v \not\in T^v_u}} \Psi^v(e, f^v).
  \end{align*}
  Note $f \not\in E_{\leq}^v$. 
  By \Cref{item:dist} of \Cref{cond:SAW-tree}, it holds that
  \begin{align} \label{eq:W-Psi-ub-1}
    \abs{W\Inf(e,f)}
    \leq \sum_{\substack{f^u \in \chi_u^{-1}(f) \\ \-{dist}_{T^u}(e, f^u) \geq g}}\abs{\Psi^u(e, f^u)} + \sum_{\substack{f^v \in \chi_v^{-1}(f) \\ \-{dist}_{T^v}(e, f^v) \geq g}} \abs{\Psi^v(e, f^v)}
  \end{align}
  By symmetry, \eqref{eq:W-Psi-ub-1} also holds for $f \in E^v_\leq$.
  Therefore, we have
  \begin{align*}
    \sum_{f:f\neq e} \abs{W\Inf(e, f)}
    &\leq \sum_{\substack{f^u \in T^u \\ \-{dist}_{T^u}(e, f^u) \geq g}}\abs{\Psi^u(e, f^u)} + \sum_{\substack{f^v \in T^v \\ \-{dist}_{T^v}(e, f^v) \geq g}} \abs{\Psi^v(e, f^v)} \\
    &\overset{(\star)}{\leq} 2C \cdot \sum_{k = g}^\infty (1 - \delta)^k
    = 2C \cdot (1-\delta)^g/\delta,
  \end{align*}
  where $(\star)$ follows from \Cref{item:DIT} of \Cref{cond:SAW-tree}.
\end{proof}

\subsubsection{The explicit form of $W\Inf$ (proof of \Cref{lem:W-Psi-explicit})}
\label{sec:W-Psi-explicit}

\begin{proposition} \label{prop:W-Psi-diff-side}
  Let $T = (V, E)$ be a tree, $\Inf$ be the influence matrix on $T$, and $W, W_u, W_v$ be the matrices defined in~\Cref{def:W}.
  For any edge $e=(u,v)$ and $f$ in $E$, it holds that
  \begin{align}
    \label{eq:f-in-Tu}
    \sum_{g \in E_u} W_u(e, g) \cdot \Inf(g, f)
    = 0&,  \quad \text{if $u$ lies on the unique path between $e$ and $f$,}\\
    \label{eq:f-not-in-Tu}
    \sum_{g \in E_u} W_u(e, g) \cdot \Inf(g, f)
    = \Psi(e,f)&, \quad  \text{if $e = f$ or $v$ lies on the unique path between $e$ and $f$.} 
  \end{align}
\end{proposition}
\begin{proof}
  We note that the condition in \eqref{eq:f-in-Tu} and \eqref{eq:f-not-in-Tu} form a partition of all the $f \in E$.
  We first prove~\eqref{eq:f-not-in-Tu}. In this case,
  $e$ lies on the unique path between $f$ and any edge $g \in E_u$.
  Therefore, by~\Cref{prop:product-distance} and~\Cref{def:W},
  \begin{align*}
    \sum_{g \in E_u} W_u(e, g) \cdot \Psi(g, f)
    &= \sum_{g \in E_u} W_u(e, g) \cdot \Psi(g, e) \cdot \Psi(e, f) \\
    &= \Psi(e, f) \cdot \widehat{W}_u\Psi_u(e, e) = \Psi(e,f).
  \end{align*}

  Now, we prove~\eqref{eq:f-in-Tu}. Since $T$ is a tree, by \Cref{prop:Cor-Q-I} and~\Cref{prop:W-Q-similarity}, $W\Inf = I$, which means $W\Inf(e, f) = 0$ by $e \neq f$.
  This indicates
  \begin{align*}
    0
    &= \sum_{g\in E_u} W_u(e, g) \Inf(g, f) + \sum_{h \in E_v} W_v(e, h) \Inf(h, f) - \Inf(e, f) 
    = \sum_{g\in E_u} W_u(e, g) \Inf(g, f),
  \end{align*}
  where the last inequality holds by~\eqref{eq:f-not-in-Tu}
\end{proof}



For graph $G=(V,E)$ and $u \in V$, recall that $T^u$ is a tree rooted at vertex $u$ defined in~\Cref{cond:SAW-tree} and $\Psi^u$ is the influence matrix on $T^u$.
Let $\widehat{W}^u_v$, $W^u_v$ and $W_u$ be defined as in~\Cref{def:W} for influence matrix $\Psi^u$ and $\Psi$.
%
Now, we are ready to prove \Cref{lem:W-Psi-explicit}.

\begin{proof}[Proof of \Cref{lem:W-Psi-explicit}]
  For $e = (u, v)$, by definition, we have
  \begin{align*}
    W\Psi(e, f) &= \sum_{g \in E_u} W_u(e, g) \Inf(g, f) + \sum_{h \in E_v} W_v(e, h) \Inf(h, f) - \Psi(e, f).
  \end{align*}
  By symmetry, it is sufficient to show that
  \begin{align} \label{eq:w-psi-explicit-target}
    \sum_{g \in E_u} W_u(e, g) \Inf(g, f) &= \sum_{\substack{f^u \in \chi_u^{-1}(f) \\ f^u \in T^u_v}}\Psi^u(e, f^u).
  \end{align}
  By \Cref{item:struct} and \Cref{item:influence} of \Cref{cond:SAW-tree}, $\Psi^u_u = \Psi_u$ and thus $\widehat{W}^u_u = \widehat{W}^u$. Therefore,
  \begin{align*}
    \sum_{g \in E_u} W_u(e, g) \Inf(g, f)
    &= \sum_{g \in E_u} W_u^u(e, g) \sum_{f^u \in \chi_u^{-1}(f)} \Psi^u(g, f^u) \\
    &= \sum_{f^u \in \chi_u^{-1}(f)} \sum_{g \in E_u} W_u^u(e, g) \cdot \Psi^u(g, f^u) \\
    &\overset{\eqref{eq:f-in-Tu}}{=} \sum_{\substack{f^u \in \chi_u^{-1}(f) \\ f^u \in T^u_v}} \sum_{g \in E_u} W_u^u(e, g) \cdot \Psi^u(g, f^u) \\
    &\overset{\eqref{eq:f-not-in-Tu}}{=} \sum_{\substack{f^u \in \chi_u^{-1}(f) \\ f^u \in T^u_v}} \Psi^u(e, f^u),
  \end{align*}
  where the last two equation follows from~\Cref{prop:W-Psi-diff-side}.
\end{proof}

\subsection{Vertex model}
For the vertex model, our main focus is the case where the underlying graph is a tree.
As we discussed before,  this is an interesting and unresolved scenario where the total influence may be unbounded yet a constant spectral independence may still exist.

Let $T=(V,E)$ be a tree rooted at $r \in V$ and  $\mu$ be a vertex model distribution  with the underlying graph $T$.
Denote the children of $u$ by $C(u)$ and the parent of $u$ by $p_u$. 
We observe that weak local correlation, with small total influences around a vertex, implies the boundedness of approximate inverse (i.e.,~\Cref{cond:boundedness}). 

\begin{lemma}\label{lem:vertex-boundedness}
  For any vertex $u \in V \setminus \{r\}$, let $\beta_u = \Cor(u,p_u)$. 
  If there exists some $\epsilon \in (0,1)$ such that for any $u \in V \setminus \{r\}$, 
  \begin{align*}
    \sum_{v \in C(u)} \beta_v^2 \le 1 - \epsilon \quad \text{and} \quad \sum_{v \in C(r)} \beta_v^2 \le \frac{1}{2\epsilon}.
  \end{align*} 
  Then,~\Cref{cond:boundedness} holds with $\alpha = \epsilon^2/4$.
\end{lemma}

\begin{remark}
  Using~\Cref{thm:Psi-lambda-max} and~\Cref{lem:vertex-boundedness}, we can reproduce the classical result of the optimal mixing of Glauber dynamics for the Ising model on trees when $\beta$ is smaller than $\beta_1(\Delta) = \frac{\sqrt{\Delta-1}+1}{\sqrt{\Delta-1}-1}$, known as the ``spin-glass critical point'', established in~\cite{Berger2005Glauber}.
  By a straightforward calculation, $\abs{\beta_u} = \abs{\Cor(u,p_u)} \le \frac{\beta - 1}{\beta + 1}$. Hence, for any $u \in V \setminus \{r\}$, 
  \begin{align*}
    \sum_{v \in C(u)} \beta_u^2 < (\Delta-1) \tp{\frac{\beta_1(\Delta) - 1}{\beta_1(\Delta) +1}}^2 = 1 \quad \text{and} \quad \sum_{v \in C(r)} \beta_u^2 < \Delta \tp{\frac{\beta_1(\Delta) - 1}{\beta_1(\Delta) +1}}^2 = \frac{\Delta}{\Delta-1}.
  \end{align*}
  Then, by~\Cref{thm:Psi-lambda-max} and~\Cref{lem:vertex-boundedness}, the Gibbs distribution $\mu$ of the Ising model on any tree $T=(V,E)$ with $\beta < \beta_1(\Delta)$ and arbitrary external field $\lambda \in \mathbb{R}_{>0}^V$ exhibits a constant spectral independence, yielding the optimal mixing result according to~\cite{chen2022localization}.
\end{remark}
The crux of the proof of~\Cref{lem:vertex-boundedness} lies in decomposing the quadratic form $\*x^T Q \*x$ of the (approximate) inverse $Q$ into a sum of squares.

\begin{proof}
  Let $\beta_{u,v} = \Cor(u,v) $ for edge $(u,v) \in E$ and $\beta_u = \beta_{u,p_u}$ for $u \in V \setminus \{r\}$. By a straightforward calculation, 
  \begin{align*}
    \widehat{Q}_e:=\tp{\Cor_e}^{-1} = 
    \frac{1}{1-\beta_{u,v}^2} \left[       
      \begin{array}{cc}   
        1 & -\beta_{u,v}\\
        -\beta_{u,v} & 1\\
      \end{array}
    \right].
  \end{align*}
  By~\Cref{def:Q-vertex}, the approximate inverse $Q$ satisfies
  \begin{equation*}
    \forall u \in E, \quad Q(u, u) = \sum_{v\in N(u)}Q_{(u,v)}(u,u)-(d_u-1)=
    \sum_{v \in N(u)} \frac{1}{1-\beta_{u,v}^2} - d_u + 1;
  \end{equation*}
  \begin{equation*}
    \forall e=(u,v)\in E, \quad Q(u, v) =Q_e(u, v)= - \frac{\beta_{u,v}}{1-\beta_{u,v}^2};
  \end{equation*}
  otherwise, $Q(u,v)=0$.
  Here, $N(u)$ is the set of neighbors of $u$ and $d_u$ is the degree of vertex $u$. We claim that the quadratic form $\*x^T Q \*x$ can be written as follows:
  \begin{align}\label{eq:quadratic}
    \*x^T Q \*x = \sum_{u \in V \setminus \{r\}} \frac{1}{1-\beta_u^2} \tp{\frac{\beta_u}{\sqrt{1-\zeta(1-\beta_u^2)}} x_{p_u} - \sqrt{1-\zeta(1-\beta_u^2)} x_u}^2 + \sum_{u \in V} \sigma(u), 
  \end{align}
  where $\zeta = \epsilon/2$ and $\sigma \in \mathbb{R}^V$ satisfies:
  \begin{align*}
    \sigma(u) := 
    \begin{cases}
      \tp{1 - \sum_{v \in C(u)} \frac{\zeta \beta_v^2}{1-\zeta (1-\beta_v^2)}} x_u^2 & \text{$u$ is the root,} \\
      \tp{\zeta - \sum_{v \in C(u)} \frac{\zeta \beta_v^2}{1-\zeta (1-\beta_v^2)}} x_u^2 & \text{otherwise.}
    \end{cases}
  \end{align*}
  Assuming~\eqref{eq:quadratic}, we conclude the proof with the following calculation:
  \begin{align*}
    \*x^T Q \*x &\ge \tp{1-\sum_{v \in C(r)} \frac{ \zeta \beta_v^2}{1-\zeta (1-\beta_v^2)}} x_r^2 + \sum_{u \in V \setminus \{r\}} \zeta \tp{1-\sum_{v \in C(u)} \frac{\beta_v^2}{1-\zeta(1-\beta_v^2)}} x_u^2  \\
    &\ge \tp{1-\epsilon \sum_{v \in C(r)} \beta_v^2} x_r^2 + \sum_{u\in V \setminus \{r\}} \zeta \tp{1-\frac{1}{1-\zeta} \sum_{v \in C(u)} \beta_v^2} x_u^2 \ge \frac{\epsilon^2}{4} \*x^T \*x, 
  \end{align*}
  where the last inequality follows from the assumption on $\sum_{v \in C(u)} \beta_v^2$.
  Therefore, it remains to verify~\eqref{eq:quadratic}. 
  It can be easily seen that the coefficients of $x_u x_v$ in both sides are identical for any $u \neq v$. Hence,
  it suffices to calculate the coefficients of $x_u^2$ term for any $u \in V$.
  
  By a straightforward calculation, the coefficients of $x_r^2$ term in the RHS of~\eqref{eq:quadratic} is
  \begin{align*}
    &\quad \sum_{v \in C(r)} \frac{\beta_v^2}{(1-\beta_v^2)(1-\zeta(1-\beta_v^2))} + \tp{1 - \sum_{v \in C(r)} \frac{\zeta \beta_v^2}{1-\zeta (1-\beta_v^2)}}\\
    &=\sum_{v \in C(r)} \frac{\beta_v^2}{1-\beta_v^2} + 1 = \sum_{v \in C(r)} \frac{1}{1-\beta_v^2} - \abs{C(r)}+ 1 = Q(r,r).
  \end{align*}
  Similarly, the coefficients of $x_u^2$ term in the RHS of~\cref{eq:quadratic} for any $u \in V \setminus \{r\}$ is given by
  \begin{align*}
    &\quad \frac{1-\zeta(1-\beta_u^2)}{1-\beta_u^2} + \sum_{v \in C(u)} \frac{\beta_v^2}{(1-\beta_v^2)(1-\zeta(1-\beta_v^2))}  +  \tp{\zeta - \sum_{v \in C(u)} \frac{\zeta \beta_v^2}{1-\zeta (1-\beta_v^2)}}\\
    &= \frac{1}{1-\beta_u^2} + \sum_{v \in C(u)} \frac{1}{1-\beta_v^2} - \abs{C(u)} = \sum_{v \in N(u)} \frac{1}{1-\beta_{u,v}^2} - d_u + 1 = Q(u,u).
  \end{align*}
  These complete the verification of~\eqref{eq:quadratic}.
\end{proof}

\section{Spectral independence in the monomer-dimer model}
\label{sec:matching-SI}

Let $\mu$ be the Gibbs distribution of the monomer-dimer model on graph $G=(V,E)$ with fugacity $\lambda > 0$,
and $\Inf \; (\Cor)$ be the (symmetrized) influence matrix of $\mu$.
The primary goal of this section is to show~\Cref{thm:matching-SI-large-girth}.
According to \Cref{thm:Psi-lambda-max}, it is sufficient for us to have the following results.

\begin{lemma} \label{lem:lambda-max-cor-Q}
  If $G$ has maximum degree $\Delta$ and girth at least $2g + 1$, then 
  \[\lambda_{\max}(Q\; \Cor) \leq 4\tp{\sqrt{1 + \lambda\Delta} + 1} \tp{1-\frac{2}{\sqrt{1+\lambda\Delta}+1}}^{g/2} + 1.\]
  In particular, if $g  \ge \tp{\sqrt{1+\lambda \Delta} + 1} \log\tp{\sqrt{1+\lambda \Delta} + 1}$, then $\lambda_{\max}(Q\; \Cor) \le 5$.
\end{lemma}

\begin{lemma} \label{lem:lambda-min-Q}
  For any graph $G$, it holds that $\lambda_{\min}(Q) \geq \frac{1}{2\lambda + 1} > 0$.
\end{lemma}

Now, we are ready to prove \Cref{thm:matching-SI-large-girth}.

\begin{proof}[Proof of \Cref{thm:matching-SI-large-girth}]
  Let $g' = \ftp{(g-1)/2}$ so that $g \geq 2g' + 1$.
  The first part follows from~\Cref{thm:Psi-lambda-max},~\Cref{lem:lambda-max-cor-Q} (take $g = g'$), and~\Cref{lem:lambda-min-Q}.
  Note that $g'/2 \geq \ftp{(g-1)/4}$.
  When $\lambda \Delta \le 3$, $\lambda_{\max}\tp{\Psi} \le 2 \lambda \Delta \le 6$ by Theorem 2.10 in~\cite{chen2021optimal}.
  Otherwise, it holds that $g' \ge 4\sqrt{\lambda \Delta} \log \tp{\lambda \Delta} \ge \tp{\sqrt{1+\lambda \Delta}+1} \log \tp{\sqrt{1+\lambda \Delta}+1}$. 
  Therefore, the second part follows from~\Cref{thm:Psi-lambda-max},~\Cref{lem:lambda-max-cor-Q}, and~\Cref{lem:lambda-min-Q}.
\end{proof}

First, we prove \Cref{lem:lambda-max-cor-Q} via the well known \emph{path-tree} for the monomer-dimer model introduced in~\cite{godsil1981matchings}.
The author proved that it can preserve the matching polynomial (i.e., the partition function of unweighted monomer-dimer model) in some sense.
As an application, the path-tree can be used to prove the roots of the matching polynomial of a graph of maximum degree $d$ are real and at most $2\sqrt{d-1}$~\cite{heilmann1972theory}.

\begin{definition}[\text{\cite{godsil1993algebratic}}] \label{def:path-tree}
  Let $G = (V, E)$ be a graph and $u \in V$ be a vertex, the \emph{path-tree} $T^u = (V^u, E^u)$ is the tree whose vertices correspond to paths in $G$ starting at $u$ and do not contain any vertex twice.
  Moreover, one path is connected to another if one extends the other by one vertex,
  and the edge used to connect them is a copy of the different edge between these two paths.
  We let $\chi_u: E^u \to E$ be the map that maps all the copied edges to their original version.
\end{definition}

From the perspective of influence, \cite{chen2021optimal} observed that there is a fine-grained relationship between influences in $G$ and influences in its path-tree.

\begin{proposition}[\text{\cite[Proposition 6.6]{chen2021optimal}}] \label{prop:path-tree-inf}
  Let $G = (V, E)$ be a graph and $u \in V$ be a vertex.
  Let $T^u$ and $\chi_u$ be defined in \Cref{def:path-tree}, and $\Psi^u$ be the influence matrix of the monomer-dimer model on $T^u$, it holds that
  \begin{align*}
    \forall e \in E_u, f \in E, \quad \Psi(e, f) &= \sum_{f' \in \chi_u^{-1}(f)} \Psi^u(e, f').
  \end{align*}
\end{proposition}
\noindent Additionally, they established the total influence decay of the monomer-dimer model.
\begin{lemma}[\text{\cite[Proposition 6.9]{chen2021optimal}}]\label{lem:inf-decay}
  For any distribution $\mu$ of monomer-dimer model on a tree with maximum degree $\Delta$ and fugacity $\lambda>0$, the influence matrix $\Psi$ of distribution $\mu$ satisfies
  \begin{align*}
    \sum_{f:\mathrm{dist}(e,f) = k} \abs{\Psi(e,f)} \le C \tp{1-\delta}^k,
  \end{align*} 
  where $\delta = 1-\sqrt{1-\frac{2}{\sqrt{1+\lambda\Delta}+1}}$ and $C = 2$.
\end{lemma}


Now, we are ready to prove \Cref{lem:lambda-max-cor-Q}.

\begin{proof}[Proof of \Cref{lem:lambda-max-cor-Q}]
  According to \Cref{def:path-tree}, \Cref{prop:path-tree-inf}, and \Cref{lem:inf-decay},  
  \Cref{item:struct,item:dist}, \Cref{item:influence} and \Cref{item:DIT} of \Cref{cond:SAW-tree} hold respectively.
  Thus, \Cref{cond:SAW-tree} holds with the value of $\delta, C$ given in \Cref{lem:inf-decay}.
  Therefore, by \Cref{lem:tree-graph}, it holds that
  \begin{align*}
    \lambda_{\max}(Q \; \Cor)
    &\le \frac{4}{1-\sqrt{1-\frac{2}{\sqrt{1+\lambda\Delta}+1}}}\tp{1-\frac{2}{\sqrt{1+\lambda\Delta}+1}}^{g/2} + 1 \\
    &\le 4\tp{\sqrt{1 + \lambda\Delta} + 1} \tp{1-\frac{2}{\sqrt{1+\lambda\Delta}+1}}^{g/2} + 1,
  \end{align*}
  where in the last inequality, we use the fact that $(1 + x)^r \leq 1 + rx$, for $x \geq -1$ and $r \in [0, 1]$.
 When $g \geq \tp{\sqrt{1 + \lambda\Delta} + 1}\log\tp{\sqrt{1 + \lambda\Delta} + 1}$,
 \begin{align*}
  \lambda_{\max}\tp{Q\Cor} &\le 4 \tp{\sqrt{1+\lambda \Delta} + 1} \tp{1-\frac{2}{\sqrt{1+\lambda \Delta}+1}}^{g/2} + 1\\
  &\le 4 \exp \tp{-\frac{g}{\sqrt{1+\lambda \Delta}+1} + \log \tp{\sqrt{1+\lambda \Delta} +1}} + 1 \le 5. \qedhere
 \end{align*}
\end{proof}

In the rest part of this section, we prove~\Cref{lem:lambda-min-Q}. 

\begin{proof}[Proof of~\Cref{lem:lambda-min-Q}]
We start the proof by evaluating entries of $\Cor_u$.
By~\Cref{prop:mat-rel}, for a vertex $u \in V$, $\Cor_u(e, e) = 1$ for every $e \in E_u$.
For distinct edges $e, f \in E_u$, it holds that
\begin{align} \label{eq:matching-sym-v}
  \Cor_u(e, f) &\overset{\text{\eqref{eq:sym-inf}}}{=} \sqrt{\Pr{e}\Pr{\overline{e}}} \cdot \Inf(e, f) \cdot \sqrt{\Pr{f}\Pr{\bar{f}}}^{-1} \overset{(\star)}{=} - \sqrt{R_e R_f},
\end{align}
where $R_e = \Pr{e}/\Pr{\overline{e}}, R_f = \Pr{f} / \Pr{\bar{f}}$ are marginal ratios of $e$ and $f$, and
$(\star)$ holds by $e \overset{u}{\sim} f$ in the monomer-dimer model sense, i.e.,
\begin{align*}
  \Inf(e, f) = - \Pr{f \mid \overline{e}}
  = - \Pr{f, \overline{e}}/\Pr{\overline{e}} 
  = - \Pr{f}/\Pr{\overline{e}}.
\end{align*}
By\eqref{eq:matching-sym-v}, $\Cor_u$ is indeed $\overline{D}_u^{-1} - \sqrt{\*r_u}\sqrt{\*r_u}^\top$, where $\overline{D_u} := \-{diag}\set{\Pr{\overline{e}}}_{e\in E_u}$
and $\sqrt{\*r_u} := (\sqrt{R_e})_{e\in E_u}$.
  
Now, by~\Cref{lem:edge-boundedness}, it is sufficient for us to show that $\lambda_{\max}\tp{\Cor_u} \le \frac{2\lambda+1}{\lambda+1}=:\beta$ for all $u \in V$, which is equivalent to
\begin{align} \label{eq:Cor-target}
  \beta I - \Cor_u = \beta I - \overline{D}_u^{-1} + \sqrt{\*r_u}\sqrt{\*r_u}^\top \succeq 0.
\end{align}

  Without loss of generality, let $E_u = \set{1, 2, \cdots, d}$ with $\mu_{1} \ge \mu_{2} \ge \ldots \ge \mu_{d}$, where $\mu_i = \Pr{i}$.
  Furthermore, denote $\Pr{\ol{i}}$ by $\mu_{\overline{i}}$.
  When $d=1$, it holds that $\beta I - \Cor_u = \beta - 1 =\frac{\lambda}{\lambda+1} \geq 0$.
  Therefore, we may assume that $d \ge 2$ throughout the proof.
  If $\beta - \frac{1}{1-\mu_{1}}\ge 0$, which is equivalent to $\mu_1\le 1-\frac{1}{\beta}$,
  $\beta I - \Cor_u \succcurlyeq \beta I - \ol{D}_u^{-1}$ would be positive semidefinite automatically.
  Therefore, we may further assume that $\mu_1 > 1-\frac{1}{\beta}$.
  We claim that
  \begin{claim}\label{claim:mu-bound}
    \minor{}If $\mu_1 > 1 - \frac{1}{\beta}$, then $\sum_{e=1}^d\mu_e \leq 2\tp{1-\frac{1}{\beta}}$.
  \end{claim}
  Assuming the correctness of~\Cref{claim:mu-bound}, it holds that $\mu_2 \leq \sum_{e=1}^d\mu_e-\mu_1<2\tp{1-\frac{1}{\beta}} - \tp{1-\frac{1}{\beta}}=1-\frac{1}{\beta}$.
  Therefore, $\beta I - \overline{D}_u^{-1}$ has exactly one negative eigenvalue.
  By~\Cref{prop:weyl} (Weyl's inequality), it is clear that $\beta I - \Cor_u$ has at most one negative eigenvalue.
  Hence, it is sufficient for us to show that
  \begin{equation*}
    \-{det}\tp{\beta I - \ol{D}_u^{-1} + \sqrt{\*r_u}\sqrt{\*r_u}^\top}\overset{(\star)}{=}\-{det}\tp{\beta I - \ol{D}_u^{-1}} \tp{1 + \sum_{e=1}^d \frac{R_e}{\beta - \mu_{\ol{e}}^{-1}}} \geq 0,
  \end{equation*}
  where $(\star)$ holds by~\Cref{prop:det-lem} (the matrix determinant lemma).
  This is equivalent to
  \begin{align} \label{eq:Q-spectral-target}
    1 + \sum_{e=1}^d \frac{R_e}{\beta - \mu_{\ol{e}}^{-1}}
    &= 1 + \sum_{e=1}^d \frac{\mu_e / (1 - \mu_e)}{(1 - x)^{-1} - (1 - \mu_e)^{-1}} 
    = 1 + (1-x)\sum_{e=1}^d \frac{\mu_e }{x - \mu_e} \leq 0
  \end{align}
  %
  %
  as $\-{det}\tp{\beta I - \ol{D}_u^{-1}} \le 0$, where $x$ is defined as $1-\frac{1}{\beta}$.
  Hence, it suffices to show that
  \begin{align} \label{eq:Q-spectral-target-1}
    \sum_{e=1}^d \frac{\mu_e}{x - \mu_e} \leq \frac{1}{x-1}.
  \end{align}
  We have the following claim.
  \begin{claim} \label{claim:Q-spectral-aux}
    If $b < x < a$, and $a + b \leq 2x$, then
    \[\frac{a}{x - a} + \frac{b}{x - b} \leq \frac{a + b}{x - (a + b)}.\]
  \end{claim}

  By~\Cref{claim:mu-bound} and our assumption that $x < \mu_1$, we have $\mu_j <x <\sum_{e=1}^{j-1} \mu_e$ for any $j\ge 2$.
  Therefore, by using~\Cref{claim:Q-spectral-aux}, it holds that $\sum_{e =1}^d \frac{\mu_e}{x-\mu_e} \le \frac{\sum_{e=1}^d \mu_e}{x-\sum_{e=1}^d\mu_e} \le \frac{1}{x-1}$, where the last inequality follows from the fact that $x < \sum_{e=1}^d \mu_e \le 1$.
  This concludes the proof.
\end{proof}


\begin{proof}[Proof of \Cref{claim:mu-bound}]
  Note that $\frac{\mu_1}{1-\sum_{e=1}^d \mu_e}$ is the marginal ratio of edge $e_1$ being chosen in an instance of the monomer-dimer model specified by graph $G=(V,E \setminus \{2,3,\ldots,d\})$ and fugacity $\lambda$.
  This indicates that $\frac{\mu_1}{1-\sum_{e=1}^d \mu_e} \leq \lambda$.
  Hence, 
  \begin{align*}
    1 - \frac{1}{\beta} < \mu_1 \le \lambda \tp{1-\sum_{e = 1}^d \mu_e},
  \end{align*}
  where the first inequality follows from our assumption on $\mu_1$. 
  Therefore, 
  \begin{align*}\sum_{e=1}^d\mu_e \leq \frac{2\lambda}{2\lambda + 1}=2\tp{1-\frac{1}{\beta}},
  \end{align*}
  where the last equation follows from the definition of $\beta$.
\end{proof}

\begin{proof}[Proof of \Cref{claim:Q-spectral-aux}]
  By a direct calculation, we have
  \begin{align*}
    \frac{a}{x-a} + \frac{b}{x-b} - \frac{a + b}{x - (a + b)} = \frac{a b (a+b-2 x)}{(a-x) (x-b) (a+b-x)} \leq 0,
  \end{align*}
  where the last inequality holds by $a + b \leq 2x$ and $b < x < a$.
\end{proof}


\subsection{Lower bound of spectral independence on graphs with parallel edges}
\label{sec:matching-SI-lb}
While we proved a constant spectral independence on graph with large girth, it is important to note that the presence of small cycles may lead to a significantly different result.
In this section, we will give a proof for \Cref{thm:matching-SI-lb}.
Specifically, we will prove that the maximum eigenvalue of influence matrix can depend on the maximum degree $\Delta$ if graph contains parallel edges, which can be seen as the case where girth is equal to $2$.


\begin{remark}
  We remark that our construction is replacing  each edge of a cycle with $\Delta/2$ parallel edges, which is exactly the $k$-transformation (with $k = \Delta/2$) on cycle defined in~\Cref{def:k-trans}.
  When the cycle is sufficiently large, the spectrum of its influence matrix is close to the influence matrix of the infinite path/cycle, whose maximum eigenvalue is $\Theta(\sqrt{\lambda})$.
  Then, our lower bound on the maximum eigenvalue can be derived from the fact that the maximum eigenvalue of the correlation matrix is preserved under the $k$-transformation (see \Cref{prop:k-trans-cor} for details).
\end{remark}

\begin{lemma} \label{lem:SI-long-cycle}
  Let $\lambda > 0$, there is a sufficiently large $n$ such that, let $C_n = ([n], E)$ be a cycle of length $n$ and $\mu$ be the Gibbs distribution of the monomer-dimer model on $C_n$ with fugacity $\lambda$.
  It holds that $\lambda_{\max}(\Inf_\mu) \geq \frac{\sqrt{\lambda}}{3}$, where $\Inf_\mu$ is the influence matrix of $\mu$.
\end{lemma}

\begin{proof}[Proof of \Cref{thm:matching-SI-lb}]
  Without loss of generality, we assume that $\Delta$ is even. 
  By \Cref{lem:SI-long-cycle}, we note that there is a sufficiently large $n$ such that $\lambda_{\max}(\Inf_\nu) \geq \frac{\sqrt{\Delta/2}}{3} \ge \frac{\sqrt{\Delta}}{5}$, 
  where $\nu$ is the Gibbs distribution for the monomer-dimer model on $C_n = ([n], E)$ with fugacity $\Delta/2$.

  Let $\mu:=\nu_{\Delta/2}$ be the distribution after doing $(\Delta/2)$-transformation on $\nu$.
  Therefore, $\mu$ is the Gibbs distribution for the monomer-dimer model on the graph that replaces each edge in $C_n$ with $\Delta/2$ parallel edges and has fugacity $1$, as desired.
  By \Cref{cor:inf-cor-lambda-max} and \Cref{prop:k-trans-cor}, it holds that
  \begin{align*}
    \lambda_{\max}(\Inf_{\mu}) \geq \lambda_{\max}(\cor_{\mu})
    =\lambda_{\max}(\cor_{\nu}) \geq \lambda_{\max}(\Inf_\nu) \cdot \min_e \nu_{\overline{e}}
    \overset{(\star)}{\geq} \lambda_{\max}(\Inf_\nu)/2 \geq \frac{\sqrt{\Delta}}{10}
  \end{align*}
  where $(\star)$ holds by the fact that $C_n$ is a cycle.
\end{proof}

Now, it only remains to prove \Cref{lem:SI-long-cycle}.
We first formalize the intuition that the entries in the influence matrix of $C_n$ converges to the corresponding entries in the influence matrix of infinite path/cycle in the following lemma.
Specifically, we study the influence between a pair of edges with distance $\ell - 1$ in $C_n$ when $n$ approaches infinity.

\begin{lemma}\label{lem:cycle-path-marignal}
  Let $\ell \in \^N$ be a constant, 
  $C_n=([n],E)$ be a cycle of length $n$ and $\lambda > 0$,
  and $\mu$ be the Gibbs distribution of monomer-dimer model with graph $C_n$ and parameter $\lambda$.
  It holds that 
  \begin{align*}
    \lim_{n\to \infty} \Inf_\mu(e_1, e_\ell) = (-R)^{\ell-1},
  \end{align*}
  where $R:=1 - \frac{2}{\sqrt{1 + 4\lambda} + 1}$, and we use $e_i$ to denote the edge $(i, i+1)$.
\end{lemma}

Now, we are ready to prove \Cref{lem:SI-long-cycle}.

\begin{proof}[Proof of \Cref{lem:SI-long-cycle}]
  Let $\ell = \ell(\lambda)$ be a constant determined later and let $\Inf_{n,\ell}$   be the principal minor of $\Inf_\mu$ on edge set ${\{e_i\}}_{i\in[\ell]}$.
  By~\Cref{prop:interlace} (the interlacing theorem), it suffices to show that $\lambda_{\max}(\Inf_{n,\ell}) \geq \frac{\sqrt{\lambda}}{3}$ for sufficiently large $n$.


  By~\Cref{lem:cycle-path-marignal} and the definition of the influence matrix,
  \begin{align}\label{eq:limit-property}
    \forall 1 \le i,j \le \ell, \quad \lim_{n \to +\infty} \Inf_{n,\ell}(e_i,e_j) = (-R)^{\abs{i-j}} =: \Inf_\ell(i, j), 
  \end{align}
  where we denote the limit as $\Inf_\ell$.
  Let $\*x = ((-1)^i)_{i\in[\ell]}$ be a test vector.
  Note that when $\ell = \ell(\lambda)$ is chosen such that $R^{\ell-1} \leq 1/2$, it holds that 
  \begin{align*} 
    \frac{\*x^\top\Inf_{\ell} \*x}{\*x^\top \*x}
    &= \frac{1}{\ell} \sum_{i=1}^\ell \sum_{j=1}^\ell \abs{\Inf_\ell(i, j)} 
    \geq \min_i \sum_{j=1}^\ell \abs{\Inf_{\ell}(i,j)}
    \overset{~\eqref{eq:limit-property}}{\geq} \sum_{j=0}^{\ell-1} R^j = \frac{1-R^{\ell-1}}{1-R} \geq \frac{1}{2(1-R)}. 
  \end{align*}
  Therefore, according to \eqref{eq:limit-property}, for sufficiently large $n$,
  \begin{align*}
    \lambda_{\max}(\Inf_{n,\ell}) &\overset{(\star)}{\geq} \frac{2}{3} \lambda_{\max}(\Inf_\ell) \geq \frac{1}{3(1-R)} \geq \frac{\sqrt{\lambda}}{3}, 
  \end{align*}
  where $(\star)$ follows from that both $\Psi_{n,\ell}$ and $\Psi_\ell$ are square matrices of a constant size $\ell$.
\end{proof}

\begin{remark}
  As a remark, a similar phenomenon for the (symmetrized) influence
  matrix of the infinite 2-regular tree $\Inf$ has been observed by some previous work~\cite{chen2021optimal, liu2023spectral}.
  Specifically, they prove that $\Inf(e, f) = (-R)^{\-{dist}(e, f)}$, where $R = 1 - \frac{2}{\sqrt{1 + 4\lambda} + 1}$ is the marginal ratio of each edge in the infinite 2-regular tree.
  Then it holds that $\lambda_{\max}(\Inf) = \norm{\Inf}_\infty = 2\sum_{i=0}^\infty R^i - 1 = 2/(1-R) - 1 = \Theta(\sqrt{\lambda})$.
\end{remark}

\begin{proof}[Proof of \Cref{lem:cycle-path-marignal}]
  Let $C_n, P_n$ be the cycle and path of length $n$, respectively.
  We note that $C_n$ have $n$ vertices $1, \cdots, n$ and $P_n$ have $n+1$ vertices $1,\cdots,n+1$.
  Without loss of generality, we assume $\ell \leq n/2$.
  According to the construction of the path-tree starting at vertex $1$, and the fact that path-tree preserves influence (\Cref{def:path-tree} and \Cref{prop:path-tree-inf}), we have
  \begin{align*}
    \Inf_{C_n}(e_1, e_\ell) = \Inf_{P_{n-1}}(e_1, e_\ell) + \Inf_{P_n}(e_1, e_{n-\ell+2}),
  \end{align*}
  where we use $\Inf_G$ to denote the influence matrix on $G$ and $e_i$ to denote the edge $(i, i+1)$.
  Hence, it holds that
  \begin{align*}
    \lim_{n\to\infty} \Inf_{C_n}(e_1, e_\ell)
    &= \lim_{n\to\infty} \Inf_{P_{n-1}}(e_1, e_\ell) + \lim_{n\to\infty} \Inf_{P_n}(e_1, e_{n-\ell+2})\\
    (\text{\Cref{lem:inf-decay}}) &= \lim_{n\to\infty} \Inf_{P_{n-1}}(e_1, e_\ell) 
    \overset{(\star)}{=} \lim_{n\to\infty} \prod_{i=1}^{\ell-1} \Inf_{P_{n-1}}(e_i, e_{i+1}) \\
    &= \lim_{n\to\infty} \prod_{i=1}^{\ell-1} \Inf_{P_{n-i}}(e_1, e_2)
    = (\lim_{n\to\infty} \Inf_{P_n}(e_1, e_2))^{\ell-1}.
  \end{align*}
  We note that the monomer-dimer model satisfies \Cref{def:conf-ind} so that $(\star)$ holds by~\Cref{prop:product-distance}.
  So, to finish the proof, it suffices for us to show that $\lim_{n\to\infty} \Inf_{P_n}(e_1, e_2) = - R$.
  For $n \geq 3$, we use $Z(P_n)$ to denote the partition function of the monomer-dimer model specified by the path $P_n$ and fugacity $\lambda$.
  Then,
  \begin{align} \label{eq:marginal-cycle}
    -\Inf_{P_n}(e_1, e_2) = \Pr{e_2\mid \ol{e_1}} = \frac{\lambda Z(P_{n-3})}{Z(P_{n-1})} = \frac{\lambda Z(P_{n-3})}{\lambda Z(P_{n-3}) + Z(P_{n-2})} = \frac{\lambda}{\lambda + \frac{Z(P_{n-2})}{Z(P_{n-3})}}.
  \end{align}
  By recursion $Z(P_n) = Z(P_{n-1}) + \lambda Z(P_{n-2})$, it holds that
  \begin{align*}
    \frac{Z(P_{n-1})}{Z(P_n)} = \frac{Z(P_{n-1})}{Z(P_{n-1}) + \lambda Z(P_{n-2})} = \frac{1}{1 + \lambda \frac{Z(P_{n-2})}{Z(P_{n-1})}}.
  \end{align*}
  Since the function $f(x) = \frac{1}{1 + \lambda x}$ has the property $\abs{(\log \circ f \circ \exp)'(x)} = \abs{\frac{\e^x \lambda}{1 + \e^x \lambda}} < 1$ for every $x \in \^R$, it holds that $\lim_{n\to \infty} \frac{Z(P_{n-1})}{Z(P_n)} = \frac{2}{\sqrt{1 + 4\lambda} + 1}$, which is the unique positive solution of the equation $x = f(x)$.
  We finish the proof by plugging this into \eqref{eq:marginal-cycle}.
\end{proof}

\section{Spectral independence in the hardcore model}
\label{sec:hardcore-SI}

In this section, we prove the spectral independence result in~\Cref{thm:hardcore-tree}. The proof of the optimal spectral gap and mixing time are deferred to~\Cref{sec:miss-proof}.
For convenience, we restate the statement of~\Cref{thm:hardcore-tree}.

\begin{theorem}[Spectral independence for the hardcore model on trees]\label{thm:hardcore-tree-SI}
  Let $T=(V,E)$ be a tree of $n$ vertices, and $0 < \lambda <(1-\delta) \e^2$ for some $\delta \in \tp{ 0, 1/10 }$.
  The Gibbs distribution $\mu$ of the hardcore model on $T$ with fugacity $\lambda$ has the  spectral independence $\lambda_{\max}(\Inf_\mu) \leq \frac{36}{\delta^2}$.
\end{theorem}

It directly follows from~\Cref{thm:Psi-lambda-max},~\Cref{lem:vertex-boundedness} and the following lemma.

\begin{lemma}\label{lem:tech-hardcore}
  Let $\mu$ be the Gibbs distribution of the hardcore model specified by a tree $T=(V,E)$ rooted at $r \in V$ of $n$ vertices, and fugacity $\lambda > 0$ with $\lambda < (1-\delta) \e^2$ for some $\delta \in (0,1/10)$. 
  For any vertex $u \in V \setminus \{r\}$, let $\beta_u:=\Cor_\mu(u,p_u)$, where $p_u$ is the parent of vertex $u$. It holds that
  \begin{align*}
    \forall u \in V,\quad \sum_{v \in C(u)} \beta_v^2 \le 1 - \frac{\delta}{3},
  \end{align*}
  where $C(u)$ is the set of children of vertex $u$.
\end{lemma}

It remains to prove~\Cref{lem:tech-hardcore}, whose proof will be given in the following subsection.

\subsection{Analysis of local influences (proof of~\Cref{lem:tech-hardcore})}
  For convenience, let $\zeta := \delta / 3$.
  For any vertices $u,v \in V$, denote  $\Pr[]{u \mid \overline{v}}$ by $\mu_u^{\overline{v}}$.
  It can be verified that for any $(u,v) \in E$, $\tp{\Cor(u,v)}^2 = \mu_u^{\overline{v}} \mu_v^{\overline{u}}$.
  Therefore, it suffices to prove that
  \begin{align}\label{eq:desire-ineq}
    \sum_{v \in C(u)} \mu_v^{\overline{u}} \mu_u^{\overline{v}} \le 1-\zeta.
  \end{align}
  First, we eliminate the term $\mu^{\ol{v}}_u$. 
  By the tree recursion~\cite{weitz2006counting}, for all $v \in C(u)$, the set of children of $u$, it holds that
  \begin{align*}
    \frac{\mu_u^{\overline{v}}}{1-\mu_u^{\overline{v}}} = \lambda \prod_{w \in N(u) \setminus \{v\} } (1-\mu_w^{\overline{u}}) \le \lambda \prod_{w \in C(u) \setminus \{v\}} (1-\mu_w^{\overline{u}}) = \frac{\lambda}{1-\mu_v^{\overline{u}}} \prod_{w \in C(u)} (1-\mu_w^{\overline{u}}). 
  \end{align*}
  Therefore,
  \begin{align*}
    \mu_u^{\overline{v}} \le \frac{\lambda\prod_{w \in C(u)} (1-\mu_w^{\overline{u}})}{1-\mu_v^{\overline{u}} + \lambda\prod_{w \in C(u)} (1-\mu_w^{\overline{u}})}.
  \end{align*}
  Plug into~\eqref{eq:desire-ineq}, it suffices to prove that
  \begin{align}\label{eq:marginal-bound-hardcore}
    \sum_{v \in C(u)} \frac{\lambda \mu^{\ol{u}}_v \prod_{w \in C(u)} (1-\mu_w^{\overline{u}})}{1-\mu_v^{\overline{u}} + \lambda\prod_{w \in C(u)} (1-\mu_w^{\overline{u}})}\le 1 - \zeta.
  \end{align}  
  We will now introduce the following lemma that transforms the multi-variable maximization problem into a univariate one.
  The proof is deferred to the end of this subsection.
\begin{lemma}\label{lem:max-hardcore}
  Let $P \in (0,1)$ and $\lambda > 0$ be fixed parameters. 
  The following achieves the maximum value when $a_1 = a_2 = \ldots = a_i$, $a_{i+1}=a_{i+2}= \ldots = a_n = 0$ for some $1 \le i \le n$:
  \begin{align*}
      \sum_{i=1}^n \frac{a_i}{1-a_i+\lambda P} \quad \text{subject to $a_i \in [0,1-P]$ and $P=\prod_{i=1}^n (1-a_i)$.}
  \end{align*}
\end{lemma}
By~\Cref{lem:tech-hardcore}, the LHS of~\eqref{eq:marginal-bound-hardcore} can be bounded by
\begin{align*}
  \sum_{v \in C(u)} \frac{\lambda \mu^{\ol{u}}_v \prod_{w \in C(u)} (1-\mu_w^{\overline{u}})}{1-\mu_v^{\overline{u}} + \lambda\prod_{w \in C(u)} (1-\mu_w^{\overline{u}})}
  \le \sup_{\substack{x \in (0,1)\\d \ge 1}} \frac{d \lambda x (1-x)^d}{1-x + \lambda (1-x)^d} =  \sup_{\substack{x \in (0,1)\\d \ge 1}} \frac{d \lambda (1-x) x^d}{x + \lambda x^d}.
\end{align*} 

It only remains to prove $F(d,x) := \frac{d\lambda (1-x) x^d}{x+\lambda x^d} \le 1-\zeta$ for all $x \in (0,1)$ and $d \ge 1$.
When $d = 1$, $F(d,x) = \frac{\lambda(1-x)}{1+\lambda} \le \frac{\lambda}{1+\lambda} < \frac{9}{10} \leq 1-\zeta$, since $\lambda \leq \e^2$ and $\zeta = \frac{\delta}{3} \leq \frac{1}{10}$.
Therefore, we may further assume that $d \ge 2$.

    For a fixed $d \geq 2$, we first investigate the value of $\sup_{x\in (0, 1)} F(d, x)$.
    Note that
    \begin{align*}
        \partial_x \log F =  \frac{d (1-x) - 1 - \lambda x^d}{(1-x) \left(\lambda  x^d+x\right)}.
    \end{align*}
    The sign of $\partial_x \log F$ is determined by $g(x) := d (1-x) - 1 - \lambda x^d$ as $x \in (0,1)$.
    Furthermore, $g(0) = d - 1 \geq 0$, $g(1) = -1 - \lambda < 0$, and $g'(x) = -d \lambda  x^{d-1}-d < 0$.
    Therefore, the equation $g(x) = 0$ has a unique solution $\hat{x}$ in $(0, 1)$ and $\sup_{x\in (0, 1)} F(d, x) = F(d, \hat{x})$.

    Since $g(\hat{x}) = 0$, $\lambda$ is uniquely determined by $\hat{x}$ and $d$ as
    \begin{align} \label{eq:lambda-sub}
        \lambda = \lambda(\hat{x}) = \frac{d(1 - \hat{x}) - 1}{\hat{x}^d}.
    \end{align}
    Plugging \eqref{eq:lambda-sub} into $F(d,x)$, we have
    \begin{align*}
        F(d, \hat{x}) = \frac{d (d (1-\hat{x})-1)}{d-1}.
    \end{align*}
    As $F(d,\hat{x})$ is monotone decreasing in $\hat{x}$, to make sure $F(d, \hat{x}) \leq 1-\zeta$, it suffices to show 
    \begin{align} \label{eq:hat-x-lb}
        \hat{x} \geq \frac{(d-1) (d+\zeta -1)}{d^2}. 
    \end{align}
    Note that the function $\lambda(\hat{x})$ in \eqref{eq:lambda-sub} is monotonically decreasing in $\hat{x}$.
    So, in order to have \eqref{eq:hat-x-lb}, we only need to make sure that
    \begin{align} \label{eq:lambda-ub}
        \lambda \leq \lambda\tp{\frac{(d-1) (d+\zeta -1)}{d^2}} = \frac{(d-1) (1-\zeta) \left(\frac{(d-1) (d+\zeta -1)}{d^2}\right)^{-d}}{d} =: G(\zeta, d).
    \end{align}
    Since $\lambda < (1 - \delta)\e^2$ and $\zeta = \delta/3$, we have $\lambda < (1 - 3\zeta)\e^2$.
    Together with~\eqref{eq:lambda-ub}, it suffices to show that $G(\zeta, d) \ge (1 - 3\zeta)\e^2$ for all $\zeta \in (0, 1/3)$ and $d \geq 2$.
    
    We note that
    \begin{align}
      G(\zeta, d)
      \nonumber &= G(0, d) \cdot (1 - \zeta) \tp{1 - \frac{\zeta}{d - 1 + \zeta}}^d \\
      \label{eq:G-lb} &\overset{(\star)}{\geq} G(0, d) \cdot (1 - \zeta) \tp{1 - \frac{d\zeta}{d-1+\zeta}} 
      \geq G(0, d) \cdot (1 - 3\zeta),
    \end{align}
    where the last inequality holds by $d \geq 2$ and $(\star)$ holds by the Bernoulli inequality (i.e., $(1 + x)^r \geq 1 + rx$, for $x \geq -1$ and $r \geq 1$) and the fact that $\frac{d \zeta}{d-1} \leq 1$.

    Note that $\lim_{d \to \infty} G(0, d) = \e^2$.
    Together with \eqref{eq:G-lb}, in order to show that $G(\zeta, d) \geq (1 - 3\zeta) \e^2$, it is sufficient for us to show that $\partial_d G(0, d) \leq 0$.
    By a straightforward calculation, the sign of $\partial_d G(0, d)$ is determined by
    \begin{align*}
      1-2 d+2 (d-1) d \log \left(\frac{d}{d-1}\right)
      \leq 1 - 2d + 2 \sqrt{d (d - 1)} \leq 0,
    \end{align*}
    where the in the first inequality, we use the fact that $\log (1 + x) \leq \frac{x}{\sqrt{x+1}}$ holds for $x \geq 0$, and the last inequality holds by the AM-GM inequality.
    This concludes the proof of~\Cref{lem:tech-hardcore}.
Finally, we prove~\Cref{lem:max-hardcore} to conclude this subsection.
\begin{proof}[Proof of~\Cref{lem:max-hardcore}]
  Let $x_1,x_2,\ldots,x_n \ge 0$ be variables such that $1 - a_i = \exp(-x_i)$ for all $1 \le i \le n$. Therefore, our goal is to maximize $\sum_{i=1}^n \frac{1-\exp(-x_i)}{\exp(-x_i) + \lambda P}$ subject to $x_i \ge 0$ for all $1 \le i \le n$ and $\sum_{i=1}^n x_i = -\log P$. 
  Note that the constraints form a closed region. Therefore, the maximum value can be achieved at a set of points $S \subseteq \mathbb{R}_{\ge 0}^n$. 
  It suffices to prove that there exists a point $\*x \in S$ with the form $x_1=x_2=\ldots=x_i$ and $x_{i+1}=x_{i+2}=\ldots=x_n=0$ for some $1 \le i \le n$.
  
  If not, let $\*x^\star \in S$ be a point with the least non-zero elements. Without loss of generality, we may assume that $x^\star_1 \ge x^\star_2 \ge \ldots \ge x^\star_n$. 
  Thus, there exists two distinct elements $x^\star_i>0$ and $x^\star_j>0$ with $x^\star_i \neq x^\star_j$. Let $A = x^\star_i+x^\star_j$, we first examine the following function defined on interval $[0,A]$:
  \begin{align*}
    f(x) = \frac{1-\exp(-x)}{\exp(-x) + \lambda P} + \frac{1-\exp(-(A-x))}{\exp(-(A-x)) + \lambda P}.
  \end{align*}
  By a straightforward calculation, the derivative of $f$ is given by 
  \begin{align*}
    f'(x) = -\frac{\exp(x) (\exp(2x) - \exp(A)) (1 + \lambda P) (-1 + \lambda^2 P^2 \exp(A) )}{(\exp(x) + \lambda P \exp(A))^2 (1 + \lambda P \exp(x))^2}
  \end{align*}
  Based on the sign of $\lambda^2 P^2 \exp(A) - 1$, there are two cases to consider.

  \textit{Case 1: $\lambda^2 P^2 \exp(A) \ge 1$.} In this case, the maximum value of $f(x)$ achieves at $x = 0$ or $x=A$. Therefore, the value of $\sum_{i=1}^n \frac{1-\exp(-x^\star_i)}{\exp(-x^\star_i) + \lambda P}$ does not decrease when $x^\star_i$ is set to $A$ and $x^\star_j$ is set to $0$. However, by doing this, the number of non-negative elements decreases by $1$, violating the least non-zero elements assumption.
  
  \textit{Case 2: $\lambda^2 P^2 \exp(A) < 1$.} In this case, the maximum value of $f(x)$ only achieves at $x = \frac{A}{2}$. Therefore, the value of $\sum_{i=1}^n \frac{1-\exp(-x^\star_i)}{\exp(-x^\star_i) + \lambda P}$ strictly increases when $x^\star_i$ and $x^\star_j$ are both set to $\frac{A}{2}$, contradicting with our assumption on maximality.
\end{proof}

\subsection{Unboundedness of spectral independence when $\lambda$ is large}
\label{sec:unbounded-hardcore}
In this subsection, we will prove \Cref{thm:unbounded-hardcore}.
%
Before proving this theorem, we first introduce a lemma in~\cite{restrepo2014phase} that guarantees the existence of a tree $T=(V,E)$ rooted at $r$ that the probability of $r$ being occupied is close to the solution of equation $x = \lambda (1-x)^{d}$.
\begin{lemma}[\text{\cite[Section 5]{restrepo2014phase}}]\label{lem:construct-fixed-point}
  Let $\lambda > 0, d \ge 3$ be constants. For any $\delta \in (0,1)$, there exists a hardcore system specified by tree $T_\delta=(V_\delta, E_\delta)$ rooted at $r$ and fugacity $\lambda$ such that
  \begin{align*}
    |\mu_r - x^\star| < \delta,
  \end{align*}
  where $\mu_r$ is the probability of $r$ being occupied, and $x^\star$ is the unique solution of $x = \lambda (1-x)^{d}$ when $x \in (0,1)$.
\end{lemma}

\begin{proof}[Proof of~\Cref{thm:unbounded-hardcore}]
  Our construction of tree $T=(V,E)$ is a complete $3$-regular tree rooted at $r$ of height $H=\left\lceil C\right\rceil$ with all its leaves substituted by $T_\delta$ for sufficiently small $\delta = \delta(C) > 0$.
  The parameters $\delta$ is chosen so that for any vertex $u \in V$ with $\mathrm{dist}(u,r)<H$ and $v \in N(u)$, the marginal probability $\mu_u^{\overline{v}}$
  satisfies that $|\mu_u^{\overline{v}} - x^\star| < \alpha$, where $x^\star$ is the solution of $x = \lambda (1-x)^3$, and $\alpha$ is a constant to be determined that only relies on $\lambda$.
  This follows from~\Cref{lem:construct-fixed-point}, the construction of $T$, and the fact that the number of pairs $u,v$ satisfying above constraints is bounded by a constant $2^{H+4}$.
  
  By a similar calculation in the proof of~\Cref{lem:vertex-boundedness}, the quadratic form of the (approximate) inverse $Q$ of $\Cor$ satisfies
  \begin{align*}
    \*x^T Q \*x = \sum_{u \in V \setminus \{r\}} \frac{1}{1-\beta_u^2} \tp{\beta_u x_{p_u} - x_u}^2 + x_r^2,
  \end{align*}
  where $r$ is the root of $T$, $p_u$ is the parent of $u$, and $\beta_u = \Cor(u,p_u) = \sqrt{\mu_u^{\overline{p_u}} \mu_{p_u}^{\overline{u}}}$.
  By~\Cref{prop:Cor-Q-I}, to show $\lambda_{\max}(\Psi) = \tp{\lambda_{\min}(Q)}^{-1} \ge C$, it suffices to prove that there exists an assignment $\*x \in \mathbb{R}^V$ satisfying
  \begin{align}\label{eq:target-lower-hardcore}
    \sum_{u \in V \setminus \{r\}} \frac{1}{1-\beta_u^2} \tp{\beta_u x_{p_u} - x_u}^2 + x_r^2 \le \frac{1}{C} \sum_{u \in V} x_u^2.
  \end{align}
  We simply assign values as follows:
  \begin{align}\label{eq:assignment-hardcore}
    x_u = 
    \begin{cases}
     1 & \text{$u$ is the root,}\\
     \beta_u x_{p_u} & \text{otherwise.}
    \end{cases}
  \end{align}
  It can be verified that the LHS of~\eqref{eq:target-lower-hardcore} is simply $x_r^2$. Therefore, it suffices to show that 
  \begin{align}\label{eq:target-lower-hardcore-2}
    \sum_{u \in V} x_u^2 \ge C x_r^2.
  \end{align}
  For all $0 \le h < H$, inductively, we will prove that
  \begin{align*}
    \sum_{\substack{u \in V\\ \mathrm{dist}(u,r) = h}} x_u^2 \ge x_r^2.
  \end{align*} 
  By the choice of $H$, equation~\eqref{eq:target-lower-hardcore-2} follows immediately.
  The induction basis, $h=0$, holds trivially. 
  For any $1 \le h < H$, suppose the claim holds for all $h'$ smaller than $h$. 
  By the assignment of vector $\*x$, 
  \begin{align*}
    \sum_{\substack{u \in V\\ \mathrm{dist}(u,r) = h}} x_u^2
    = \sum_{\substack{v \in V\\ \mathrm{dist}(v,r) = h-1}} \tp{\sum_{u \in C(v)} \mu_u^{\overline{v}} \mu_v^{\overline{u}}} x_v^2.
  \end{align*}
  Therefore, by induction hypothesis, it only remains to show that $\sum_{v \in C(u)} \mu_u^{\overline{v}} \mu_v^{\overline{u}} \ge 1$. 
  By our previous construction of the tree and assumption on marginal probability, $\mu_u^{\overline{v}} \mu_v^{\overline{u}} \ge (x^\star - \alpha)^2$ and $|C(u)|\ge 2$. Therefore, it suffices to show that $2(x^\star-\alpha)^2 \ge 1$, i.e., $x^\star \ge \alpha + \frac{1}{\sqrt{2}}$.
  Recall $x^\star$ is the unique solution of $x = \lambda (1-x)^3$. Therefore, it suffices to show that 
  \begin{align*}
      \alpha + \frac{1}{\sqrt{2}} \le \lambda \tp{1-\frac{1}{\sqrt{2}} - \alpha}^3.
  \end{align*}
  When $\lambda > \tp{1-\frac{1}{\sqrt{2}}}^{-3} \cdot \frac{1}{\sqrt{2}} \approx 28.14$, there must exist a sufficiently small $\alpha$ satisfying the constraint.
\end{proof}

\section{Fast mixing of Glauber dynamics}\label{sec:miss-proof}

In this section, we will prove the spectral gap result of \Cref{thm:matching-tree} and \Cref{thm:hardcore-tree}.
The proof relies on several results on the mixing time or spectral gap of Glauber dynamics via spectral independence.
Before introducing these results, we first introduce several definitions.

\begin{definition}[tilted distribution~\cite{chen2021rapid}]
  Let $\mu$ be a distribution over $2^U$ and $\theta > 0$. The tilted distribution $\theta * \mu$ is given by
  \begin{align*}
    \forall S \subseteq U,\quad (\theta * \mu)(S) \propto \mu(S) \theta^{\abs{S}}.
  \end{align*}
\end{definition}

\begin{definition}[marginal bound~\cite{chen2021optimal}]
  Let $\mu$ be a distribution over $2^U$ and $b > 0$ be a real number. The distribution $\mu$ is $b$-marginally bounded if $\min\set{\mu_u,1-\mu_u} \ge b$ for all $u \in U$, where $\mu_u$ denotes the probability $\Pr[\mu]{u}$.
\end{definition}

We introduce several results based on spectral independence.
\begin{lemma}[\cite{chen2021optimal}]\label{lem:optimal-mixing}
  Let $\mu$ be a Gibbs distribution of a spin system specified by graph $G=(V,E)$ with maximum degree $\Delta$. 
  If $\mu$ is $\eta$-spectrally independent and $b$-marginally bounded under all pinnings, then the mixing time of Glauber dynamics on $\mu$ is bounded by $\tp{\frac{\Delta}{b}}^{O(\eta/b^2 + 1)} n \log n$.
\end{lemma}

\begin{lemma}[\cite{chen2021rapid, chen2022localization}]\label{lem:boost-SI}
  Let $\mu$ be a Gibbs distribution on spin system on graph $G=(V,E)$ with maximum degree $\Delta$. If $(\lambda * \mu)$ is $\eta$-spectrally independent under all pinnings and $\lambda \in (0,1)$, then for any $\theta \in (0,1)$,
  \begin{align*}
    \lambda^{\mathrm{GD}}_{\mathrm{gap}}(\mu) \ge \theta^{O(\eta)} \lambda^{\mathrm{GD}}_{\mathrm{min-gap}}(\theta * \mu),
  \end{align*}
  where $\lambda^{\mathrm{GD}}_{\mathrm{gap}}(\mu)$ is the spectral gap of Glauber dynamics on $\mu$, and $\lambda^{\mathrm{GD}}_{\mathrm{min-gap}}(\theta * \mu)$ is the minimum spectral gap of $\theta * \mu$ over all possible pinnings.
\end{lemma}

\begin{remark}
In~\cite[Theorem 60]{chen2022localization}, a modified log-Sobolev constant version of~\Cref{lem:boost-SI} was proved, assuming entropic independence of distribution $(\lambda * \mu)$ under all pinnings. 
  The spectral gap version stated here can follow from the same abstract frmework.
\end{remark}

For the hardcore model, it has been shown that the Glauber dynamics on tree has an optimal spectral gap when fugacity $\lambda$ is small enough.

\begin{lemma}[\cite{efthymiou2023optimal}]\label{lem:SI-hardcore-small}
  Let $\mu$ be a Gibbs distribution for hardcore model on tree with fugacity $\lambda < 1.1$. 
  The spectral gap of Glauber dynamics is at least $\frac{1}{Cn}$ for some constant $C$.
\end{lemma}

By a similar argument, we can show that the Glauber dynamics for monomer-dimer model on trees also has an optimal spectral gap when fugacity $\lambda$ is small enough.

\begin{lemma}\label{lem:easy-matching}
  Let $\mu$ be a Gibbs distribution of the monomer-dimer model on an arbitrary tree with fugacity $\lambda \le 0.1$. The spectral gap of Glauber dynamics is at least $\frac{1}{Cn}$ for some constant $C$. 
\end{lemma}

The proof of~\Cref{lem:easy-matching} is deferred to~\Cref{sec:sg-matching}.
We are ready to prove the spectral gap result of~\Cref{thm:matching-tree} and~\Cref{thm:hardcore-tree}.
\begin{proof}[Proof of the spectral gap result of~\Cref{thm:matching-tree} and~\Cref{thm:hardcore-tree}]
  The optimal spectral gap for monomer-dimer model on tree with constant fugacity $\lambda>0$ follows from~\Cref{lem:easy-matching},~\Cref{lem:boost-SI} and~\Cref{thm:matching-SI-large-girth} by taking $\theta = \min\set{1,\frac{1}{10\lambda}}$ in~\Cref{lem:boost-SI}.
  
  Similarly, the optimal spectral gap for hardcore model on tree with $\lambda < (1-\delta) \e^2$ follows from~\Cref{lem:SI-hardcore-small},~\Cref{lem:boost-SI} and the spectral independence result of~\Cref{thm:hardcore-tree} by taking $\theta = \e^{-2}$ in~\Cref{lem:boost-SI}.
  
  Note that for each vertex $u \in V$, 
  \begin{align*}
    \frac{\lambda}{1+\lambda} \tp{\frac{1}{1+\lambda}}^\Delta \le \mu_u \le \frac{\lambda}{1+\lambda}.
  \end{align*}
  When $\lambda\le \frac{1}{2\Delta}$, it is known that Glauber dynamics mixes in $O(n \log n)$. When$\frac{1}{2\Delta} < \lambda < \e^2$, $\min \set{\mu_u,1-\mu_u} \ge \exp(-10\Delta)$. Therefore, by~\Cref{lem:optimal-mixing}, the Glauber dynamics mixes in $O_{\Delta,\delta}(n \log n)$ when maximum degree $\Delta$ is constant.
\end{proof}

\subsection{Optimal spectral gap for monomer-dimer model on tree}\label{sec:sg-matching}

Before delving into the proof of~\Cref{lem:easy-matching}, we introduce the approximate tensorization of variance.

\begin{definition}[local variance]
  Let $\mu$ be a distribution over $2^U$. For any subset $S \subseteq U$ and function $f:2^U \to \mathbb{R}$, 
  $\mu_S(f)$ is a function supported on $2^U$ and defined as
  \begin{align}\label{eq:def-mu-S}
    \forall W \subseteq U, \quad [\mu_S(f)](W)
    = \E[R \sim \mu]{f(R) \mid R \cap \overline{S} = W \cap \overline{S}},
  \end{align}
  where $\overline{S}$ is the complement of $S$. Furthermore, define the local variance $\Var[S]{f}$ as follows:
  \begin{align}\label{eq:def-var-S}
    \Var[S]{f} = \mu_S(f^2) - \tp{\mu_S(f)}^2.
  \end{align}
  When $S = U$, $\mu_S(f)$ and $\Var[S]{f}$ is a constant, and we omit the script in this case. 
  If $S = \{v\}$, we will write $\mu_v(f)$ and $\Var[v]{f}$ instead for simplicity.
\end{definition}
\begin{definition}[approximate tensorization of variance~\cite{caputo2015approximate}]
  Let $\mu$ be a distribution over $2^U$. The distribution satisfies $C$-approximate tensorization of variance for some constant $C > 0$, if for all $f:2^U \to \mathbb{R}$,
  \begin{align*}
    \Var[]{f} \le C \sum_{u \in U} \mu[\Var[u]{f}].
  \end{align*}
\end{definition}

It is known that the approximate tensorization of variance relates closely to the spectral gap of Glauber dynamics.
\begin{proposition}[\cite{caputo2015approximate}]\label{prop:tensor-sg}
  Let $\mu$ be a distribution over $2^U$. If the distribution satisfies $C$-approximate tensorization of variance for some constant $C>0$, then the spectral gap of Glauber dynamics on $\mu$ is at least $\frac{1}{C|U|}$.
\end{proposition}

Finally, we introduce several properties of the local variance.
\begin{proposition}[\cite{efthymiou2023optimal}]\label{eq:smaller}
  Let $\mu$ be a Gibbs distribution over $2^E$ on graph $G=(V,E)$. For edge subset $S,T \subseteq E$ with $\mathrm{dist}(S,T) \ge 2$, it holds that
  \begin{align*}
    \mu[\Var[S]{\mu_T[f]}] \le \mu[\Var[S]{f}].
  \end{align*}
\end{proposition}

\begin{proposition}[law of total variance]\label{eq:total-variance}
  Let $\mu$ be a Gibbs distribution over $2^E$ on graph $G = (V, E)$.
  For any $f: \Omega \to \^R$, and $S\subseteq E$, it holds that
  \begin{align*}
    \Var{f} &= \mu[\Var[S]{f}] + \Var{\mu_S[f]}.
  \end{align*}
\end{proposition}

By~\Cref{prop:tensor-sg}, to prove~\Cref{lem:easy-matching}, it suffices to show the following lemma.
\begin{lemma}\label{lem:easy-matching-2}
  Let $\mu$ be a Gibbs distribution for monomer-dimer model on tree $T=(V,E)$ rooted at a degree $1$ vertex $r \in V$ with fugacity $\*\lambda \in \mathbb{R}_{>0}^{E}$.
  If $\lambda_e \le 0.1$ for all $e \in E$, then for all $f \in \mathbb{R}^E$,
  \begin{align*}
    \Var[]{f} \le \sum_{e \in E} F_{T,e}(\lambda_e) \mu[\Var[e]{f}],
  \end{align*}
  where $F_{T,e}(x)$ is defined as
  \begin{align}\label{eq:def-F}
    F_{T,e}(x) = 
    \begin{cases}
      3(1+x) & \text{one of the endpoints in $e$ is a leaf,}\\
      6(1+x) & \text{otherwise.} 
    \end{cases}
  \end{align}
\end{lemma}

\begin{proof}
The proof of~\Cref{lem:easy-matching-2} follows from the method in~\cite{efthymiou2023optimal}.
  We will prove by induction on the size of tree $T=(V,E)$ rooted at a degree $1$ vertex $r$.
  The induction basis, where $|E|=1$, holds trivially. For any tree $T$ with size $|E| > 1$, there exists a vertex $u \in V \setminus \{r\}$ such that $|N(u)| > 1$ and all children of $u$ are leaves, as the degree of root $r$ is $1$.
  Let $H=\{h_1,h_2,\ldots,h_d\}$ be the edges connecting $u$ and its children, and $e$ be the edge connecting $u$ and its parent.
  We define the tree $T^\star=(V^\star,E^\star)$ as tree $T$ by removing all children of $u$, and $\*\lambda^\star \in \mathbb{R}_{>0}^{E^\star}$ as follows:
  \begin{align}\label{eq:lambda-star}
    \forall g \in E^\star,\quad  \lambda^\star_g = 
    \begin{cases}
      \frac{\lambda_e}{1+\lambda_H} & g = e,\\
      \lambda_g & \text{otherwise,}
    \end{cases}
  \end{align}
  where $\lambda_H = \sum_{i=1}^d \lambda_{h_i}$ for simplicity of notation.
  Finally, let $\mu^\star$ be the Gibbs distribution for monomer-dimer model on tree $T^\star$ with fugacity $\*\lambda^\star$.
  It can be verified that $\mu^\star$ is the marginal distribution of $\mu$ on $E \setminus H$, i.e.,
  \begin{align}\label{eq:rel-mu-star}
    \forall S \subseteq E^\star, \quad \mu^\star(S) = \mu(S) + \sum_{i=1}^d \mu(S \cup \{h_i\}) =
    \begin{cases}
      \mu(S)& e \in S,\\
      \tp{1+\lambda_H} \mu(S) & \text{otherwise.}
    \end{cases}
  \end{align} 
  For any function $f:2^E \to \mathbb{R}$, define $f^\star:2^{E^\star} \to \mathbb{R}$ as follows:
  \begin{align}
    \nonumber \forall S \subseteq E^\star, \quad f^\star(S) = \mu_H[f](S) &= \frac{\mu(S) f(S) + \sum_{i=1}^d \mu(S \cup \{h_i\}) f(S \cup \{h_i\})}{\mu(S) + \sum_{i=1}^d \mu(S \cup \{h_i\})}\\
    \label{eq:rel-f-star} &= 
    \begin{cases}
      f(S) & e \in S,\\
      \frac{f(S) + \sum_{i=1}^d \lambda_{h_i} f(S \cup \{h_i\})}{1+\lambda_H} & \text{otherwise.}
    \end{cases} 
  \end{align} 
  
  By~\Cref{eq:total-variance} and the induction hypothesis, it holds that
  \begin{align*}
    \Var[\mu]{f} &= \mu[\Var[H]{f}] + \Var[]{\mu_H(f)}\\
    &=\mu[\Var[H]{f}] + \Var[]{f^\star}\\
    &\le \mu[\Var[H]{f}] + \sum_{g \in E^\star} F_{T^\star,g}(\lambda^\star_g) \mu^\star[\Var[g]{f^\star}]\\
    &\le \mu[\Var[H]{f}] + \sum_{g \in E^\star \setminus \{e\}} F_{T,g}(\lambda_g) \mu[\Var[g]{f}] + F_{T^\star,e}(\lambda^\star_e) \mu^\star[\Var[e]{f^\star}],
  \end{align*}
  where the last inequality follows from~\Cref{eq:smaller} that $\mu^\star[\Var[g]{f^\star}] = \mu[\Var[g]{\mu_H[f]}] \le \mu[\Var[g]{f}]$ and $F_{T^\star,g} = F_{T,g}$ for all edge $g \in E^\star \setminus \{e\}$.
  Therefore, it suffices to show that
  \begin{align}\label{eq:require-ineq}
    \mu[\Var[H]{f}] + F_{T^\star,e}(\lambda_e^\star) \mu^\star[\Var[e]{f^\star}] \le F_{T,e}(\lambda_e) \mu[\Var[e]{f}] + \sum_{i=1}^d F_{T,h_i}(\lambda_{h_i})\mu[\Var[h_i]{f}]. 
  \end{align} 

  In order to prove \eqref{eq:require-ineq}, we have the following observations.
  \begin{proposition} \label{prop:var-decompose}
    For every $f : \Omega \to \^R$, it holds that
    \begin{align}
      \label{eq:decompose-H} \mu[\Var[H]{f}]
      &\leq \sum_{i=1}^d 2(1 + \lambda_{h_i}) \mu[\Var[h_i]{f}] \\
      \label{eq:decompose-e} \mu^\star[\Var[e]{f^\star}]
      &\leq \frac{2}{1 + \lambda^\star_e} \tp{(1 + \lambda_e)\mu[\Var[e]{f}] + \sum_{i=1}^d \lambda_e\mu[\Var[h_i]{f}]}
    \end{align}
  \end{proposition}
  \Cref{prop:var-decompose} is obtained by a straightforward calculation, and its proof is given at the last part of this section.
  According to \Cref{prop:var-decompose}, we have
  \begin{align*}
    \mu[\Var[H]{f}] &+ F_{T^\star, e}(\lambda^\star_e) \mu^\star[\Var[e]{f^\star}] \\
    &\leq F_{T^\star, e}(\lambda^\star_e) \frac{2(1 + \lambda_e)}{1 + \lambda^\star_e} \mu[\Var[e]{f}] + \sum_{i=1}^d \tp{F_{T^\star, e}(\lambda^\star_e) \frac{2\lambda_e}{1 + \lambda^\star_e} + 2(1 + \lambda_{h_i})} \mu[\Var[h_i]{f}] \\
    &\overset{\eqref{eq:def-F}}{=} 6(1 + \lambda_e) \mu[\Var[e]{f}] + \sum_{i=1}^d(6\lambda_e + 2(1 + \lambda_{h_i})) \mu[\Var[h_i]{f}]\\
    &\overset{\lambda_e \leq 0.1}{\leq} 6(1 + \lambda_e) \mu[\Var[e]{f}] + \sum_{i=1}^d 3(1 + \lambda_{h_i}) \mu[\Var[h_i]{f}]\\
    &\overset{\eqref{eq:def-F}}{=} F_{T,e}(\lambda_e) \mu[\Var[e]{f}] + \sum_{i=1}^d F_{T,h_i}(\lambda_{h_i})\mu[\Var[h_i]{f}].
  \end{align*}
  This proves \eqref{eq:require-ineq} and finishes the proof.
 \end{proof}


  Now we only left to prove \Cref{prop:var-decompose}, which is indeed \eqref{eq:decompose-H} and \eqref{eq:decompose-e}.
  In order to do so, we will calculate $\mu[\Var[H]{f}]$, $\mu^\star[\Var[e]{f^\star}]$, $\mu[\Var[h_i]{f}]$, $\mu[\Var[e]{f}]$, respectively.
  Let $\+I_T$ be the set of matchings in $T$. By a straightforward calculation,
  \begin{align}
    \nonumber \mu[\Var[e]{f}] &= \sum_{S \subseteq E \setminus H\setminus \set{e}} \frac{\mu(S)\mu(S \cup \{e\})}{\mu(S) + \mu(S \cup \{e\})} \tp{f(S) - f(S \cup \{e\})}^2\\
    \label{eq:var-e} &=\frac{\lambda_e}{1+\lambda_e}\sum_{\substack{S\subseteq E\setminus H\setminus \set{e}\\ S \cup \{e\} \in \+I_T}} \mu(S) \tp{f(S) - f(S \cup \set{e})}^2, \\
    \nonumber \mu[\Var[h_i]{f}] &= \frac{\lambda_{h_i}}{1+\lambda_{h_i}} \sum_{S\subseteq E\setminus H\setminus \set{e}} \mu(S) \tp{f(S) - f(S \cup \{h_i\} )}^2 \\
    \label{eq:var-hi}&\ge \frac{\lambda_{h_i}}{1+\lambda_{h_i}} \sum_{\substack{S\subseteq E\setminus H\setminus \set{e}\\ S \cup \{e\} \in \+I_T}} \mu(S) \tp{f(S) - f(S \cup \{h_i\} )}^2,
  \end{align}

  \begin{proof}[Proof of \eqref{eq:decompose-H}]
    For simplicity, let $\{h_0\} = \emptyset$ and $\lambda_{h_0} = 1$. By definition,
    \begin{align*}
      \mu[\Var[H]{f}] &= \frac{1}{2} \sum_{S\subseteq E^\star\setminus\set{e}} \mu^\star(S) \sum_{i=0}^d \sum_{j=0}^d \frac{\mu(S \cup \{h_i\})}{\mu^\star(S)} \frac{\mu(S \cup \{h_j\})}{\mu^\star(S)} \tp{f(S \cup \{h_i\}) - f(S \cup \{h_j\})}^2\\
                      &\overset{\eqref{eq:rel-mu-star}}{=} \frac{1}{2(1+\lambda_H)} \sum_{S\subseteq E\setminus H\setminus \set{e}} \mu(S) \sum_{0 \le i,j \le d} \lambda_{h_i} \lambda_{h_j} \tp{f(S \cup \{h_i\}) - f(S \cup \{h_j\})}^2\\
    &\overset{(*)}{\le} 2 \sum_{S\subseteq E\setminus H\setminus \set{e}} \mu(S) \sum_{i=1}^d \lambda_{h_i} \tp{f(S \cup \{h_i\}) - f(S)}^2\\
    &= \sum_{i=1}^d 2(1+\lambda_{h_i}) \frac{\lambda_{h_i}}{1+\lambda_{h_i}} \sum_{S\subseteq E\setminus H\setminus \set{e}} \mu(S) \tp{f(S) - f(S \cup \{h_i\} )}^2\\
    &\overset{(+)}{=} \sum_{i=1}^d 2(1+\lambda_{h_i}) \mu[\Var[h_i]{f}],
  \end{align*}
  where $(+)$ follows from \eqref{eq:var-hi} and $(*)$ follows from that
  \begin{align*}
    \tp{f(S \cup \{h_i\}) - f(S \cup \{h_j\})}^2 &\le 2 \tp{f(S \cup \{h_i\}) - f(S)}^2 + 2 \tp{f(S \cup \{h_j\}) - f(S)}^2. \qedhere
  \end{align*}
\end{proof}
\begin{proof}[Proof of \eqref{eq:decompose-e}]
  Similar to the proof of \eqref{eq:decompose-H}, we have
  \begin{align}
    \nonumber 
    \mu^\star[\Var[e]{f^\star}]
    &= \frac{\lambda^\star_e}{1+\lambda^\star_e}\sum_{\substack{S\subseteq E^\star\setminus \set{e}\\ S \cup \{e\} \in \+I_{T^\star} }} \mu^\star(S) \tp{f^\star(S) - f^\star(S \cup \set{e})}^2\\
    \label{eq:decompose-e-mid}
    {\scriptsize \text{by \eqref{eq:lambda-star}\eqref{eq:rel-mu-star}\eqref{eq:rel-f-star}}}
    &=\frac{\lambda_e}{1 + \lambda_e^\star} \sum_{\substack{S\subseteq E\setminus H\setminus \set{e}\\ S \cup \{e\} \in \+I_T}} \mu(S) \tp{\sum_{i=1}^d\frac{\lambda_{h_i}}{1 + \lambda_H}(f(S\cup h_i) - f(S)) + f(S) - f(S\cup \set{e})}^2.
  \end{align}
  By applying Cauchy's inequality twice, we have
  \begin{align}
    \nonumber 
    &\tp{\sum_{i=1}^d\frac{\lambda_{h_i}}{1 + \lambda_H}(f(S\cup \set{h_i}) - f(S)) + f(S) - f(S\cup \set{e})}^2 \\
    \nonumber 
    &\quad \leq 2 \tp{\sum_{i=1}^d\frac{\lambda_{h_i}}{1 + \lambda_H}(f(S\cup \set{h_i}) - f(S))}^2 + 2\tp{f(S) - f(S \cup \set{e})}^2 \\
    \label{eq:decompose-e-cauchy}
    &\quad \leq 2 \sum_{i=1}^d\frac{\lambda_{h_i}}{1 + \lambda_H}(f(S\cup \set{h_i}) - f(S))^2 + 2\tp{f(S) - f(S \cup \set{e})}^2,
  \end{align}
  where the last inequality follows from Cauchy's inequality and the fact that $\sum_{i=1}^d \frac{\lambda_{h_i}}{1 + \lambda_H} \leq 1$.

  Combining \eqref{eq:decompose-e-mid} and \eqref{eq:decompose-e-cauchy}, we have
  \begin{align}
    \mu^\star[\Var[e]{f^\star}]
    \nonumber 
    &\leq \frac{2\lambda_e}{1 + \lambda^\star_e} \sum_{\substack{S\subseteq E\setminus H\setminus \set{e}\\ S \cup \{e\} \in \+I_T}} \mu(S) (f(S) - f(S\cup \set{e}))^2 \\
    \label{eq:decompose-e-final}
    &\quad + \frac{2\lambda_e}{1 + \lambda^\star_e} \sum_{i=1}^d \frac{\lambda_{h_i}}{1 + \lambda_{h_i}} \sum_{\substack{S\subseteq E\setminus H\setminus \set{e}\\ S \cup \{e\} \in \+I_T}} \mu(S) (f(S \cup \set{h_i}) - f(S))^2,
  \end{align}
  where we use the fact that $\lambda_{h_i} \leq \lambda_H$.
  Finally, we note that \eqref{eq:decompose-e} could be proved by combining \eqref{eq:decompose-e-final}, \eqref{eq:var-e}, and \eqref{eq:var-hi}.
\end{proof}

\bibliographystyle{alpha}
\bibliography{refs}

\end{document}